\documentclass{article}  
\usepackage[margin=1 in]{geometry}
\usepackage{authblk,amsthm}
\usepackage{blindtext}
\usepackage{hyperref}
\usepackage{graphicx}
\usepackage{amsfonts} 
\usepackage{pifont}
\newcommand{\IR}{\mathbb{R}}
\def\s{{\mathcal S}}
\def\OO{{\cal O}} 
\def\Q{{\cal Q}}
\def\A{{\cal A}}

\def\I{{\mathcal I}}
\def\K{{\mathcal K}}
\def\P1{{\mathcal P_1}}
\def\P2{{\mathcal P_2}}
\def\P{{\mathcal P}}
\def\B{{\mathcal B}}

\def\QH{{\mathcal Q(\sigma)}}

 \usepackage{mathtools}
 \usepackage{pdfpages}
\usepackage{caption}
\usepackage{subcaption}
\DeclareMathOperator{\NC}{\textsc{Near-Center}}
\DeclareMathOperator{\RIR}{\textsc{Randomized-Iterative-Reweighting}}

\newtheorem{corollary}{Corollary}
\newtheorem{clm}{Claim}

\newtheorem{theorem}{Theorem}
\newtheorem{lemma}{Lemma}

\DeclareMathOperator{\BPA}{\textsf{BestPoint-Algorithm}}

\usepackage{color}

\newcommand{\blue}{\textcolor{blue}}

\begin{document}
\title{Online Hitting of Unit Balls and  Hypercubes in $\IR^d$ using  Points from~$\mathbb{Z}^d$\thanks{Preliminary version of this paper appeared in the 28th international computing and combinatorics conference (COCOON), 2022~\cite{DeS22}.
   }}

\author[1]{Minati De\footnote{Partially supported by SERB-MATRICS grant MTR/2021/000584.}}
\affil{Deptartment of Mathematics\\ Indian Institute of  Technology Delhi, India\\
\texttt{\{minati,satyam.singh$^{\star}$\}@maths.iitd.ac.in}}
 
\author[1]{Satyam Singh$^{\star}$\footnote{Supported by CSIR (File Number-09/086(1429)/2019-EMR-I).

$\ ^{\star}$Corresponding author}}

\maketitle          
\begin{abstract}
We consider the online hitting set problem for the range space $\Sigma=(\cal X,\cal R)$, where the point set $\cal X$ is known beforehand, but the set $\cal R$ of geometric objects is not known in advance.
Here, objects {from $\cal R$} arrive one by one. The objective {of the problem} is to maintain a hitting set of the minimum cardinality by taking irrevocable decisions. 
In this paper, we {consider} the problem when objects are unit balls or unit hypercubes in $\IR^d$, and the points from $\mathbb{Z}^d$ are used  for hitting them.
{First, we address the case} when objects are unit intervals in $\mathbb{R}$ and present an optimal deterministic algorithm with competitive ratio of~$2$.
Then, {we consider the case} when objects are unit balls. For hitting unit balls in $\IR^2$ and $\IR^3$, we present $4$ and $14$-competitive deterministic algorithms, respectively. On the other hand, for hitting unit balls in $\IR^d$, we propose a $O(d^4)$-competitive deterministic algorithm, and { we demonstrate that}, for $d<4$,  the competitive ratio of any deterministic algorithm is at least $d+1$.
In the end, we {explore the case where} objects are unit hypercubes. For hitting unit hypercubes in $\IR^2$ and $\IR^3$, we obtain $4$ and $8$-competitive deterministic algorithms, respectively. For hitting unit hypercubes in $\IR^d$ ($d\geq 3$), we present a $O(d^2)$-competitive randomized algorithm. {Furthermore,} we prove that the competitive ratio of any deterministic algorithm for the problem is at least $d+1$ for any  $d\in\mathbb{N}$.

\end{abstract}

\section{Introduction}
The hitting set problem and the set cover problem are one of the most fundamental problems in combinatorial optimization~\cite{Alon09,EvenS14,Feige98,amitkumar,HochbaumM85}.
Let  $\Sigma=(\cal X,\cal R)$  be a \emph{range space} where $\cal X$ is a set of \emph{elements} and $\cal R$ is a family of subsets of $\cal X$ called \emph{ranges}. A subset $\cal H \subseteq \cal X$ is called a \emph{hitting set} of the range space $\Sigma$ if the set $\cal H$ intersects every range $r$ in $\cal R$ and a subset $\cal C \subseteq \cal R$ is called a \emph{set cover} of the range space $\Sigma$ if the union of ranges in $\cal C$ covers all elements of $\cal X$. The aim of  the hitting set (respectively, set cover) problem is to find a hitting set $\cal H$ (respectively, set cover $\cal C$)  of the minimum cardinality. 
It is well known that a set cover of $\Sigma =({\cal X}, {\cal R})$ is a hitting set of the dual range space $\Sigma^{\perp}=({\cal X}^{\perp}, {\cal R}^{\perp})$. Here, for each range $r\in{\cal R}$ there is an element in ${\cal X}^{\perp}$ and for each element $x\in{\cal X}$ there is a range $r_x$, namely, $r_x=\{r\in{\cal R}\ |\ x\in r\}$, in ${\cal R}^{\perp}$~\cite{AgarwalP20}.

 Due to numerous applications in wireless sensor networks, VLSI design, resource allocation and databases,   researchers have considered the set $\cal X$ to be a collection of points from $\IR^d$ and $\cal R$ to be a finite family of geometric objects chosen from some infinite class (hypercubes, balls, etc.)~\cite{AgarwalP20,ChanH20,FriederichGGHS23,Ganjugunte11,MegiddoS84,MustafaR10}. In this case, ranges are ${\cal X} \cap r$ for any object $r \in\cal R$. With a slight misuse of the notation, we will use $\cal R$ to signify both the set of ranges as well as the set of objects that define these ranges. A \emph{geometric range space} $\Sigma=(\cal X, \cal R)$ consists of a point set $\cal X$ containing points and a set $\cal R$ {is a family of} geometric objects. The \emph{geometric hitting set} problem is to find the minimum number of points from $\cal X$ to hit all the objects in $\cal R$. The \emph{geometric set cover} problem is to find the minimum number of objects in $\cal R$ that covers all the points in $\cal X$.

For the geometric hitting set problem in an online setting, the point set $\cal X$ is known beforehand, but the set $\cal R$ of geometric objects is not known in advance.
Here, objects from $\cal R$ arrive one by one. 
An online algorithm needs to maintain a feasible hitting set $\cal H$ for the already arrived  objects. Upon the arrival of a new  object $\sigma\in\mathcal{R}$, if $\sigma$ does not contain any point from the existing hitting set $\cal H$, the algorithm needs to add a point $p\in\cal X$ to $\cal H$ to hit $\sigma$.
The decision to add a point to the solution set is  irrevocable, i.e., the online algorithm can not remove any point from the existing hitting set in future. 
{Due to the result of Even and Smorodinsky\cite{EvenS14}, we know that no online algorithm can obtain a competitive ratio better than $\Omega(\log n)$ for hitting $n$ intervals in the range $[1,n]$ using points $\P=\{1,2,\ldots,n\}$. Due to this pessimistic result, in this paper, we consider the geometric hitting set problem in an online setting, where ${\cal X}=\mathbb{Z}^d$ and the set $\cal R$ is a finite family of translates of an object $\sigma^*$ in $\IR^d$.}
For simplicity, we will use the term online hitting set problem (respectively, online covering problem) instead of geometric hitting set problem in the online setup (respectively, geometric set cover problem in the online setup).

{One real-life application of the hitting set problem is as follows. Let us consider a planned city where one can install base stations at specific locations from a rectilinear grid. 
Here, points represent base stations, and ranges represent objects centred at clients. The clients are coming one by one and upon the arrival of an uncovered client, from any location in the city, the algorithm must select a base station serving it.
The objective is to minimize the number of base stations.  Since installing a base station is expensive, the decision is considered to be irrevocable.}

We use competitive analysis to analyze the quality of our online algorithm~\cite{BorodinE}. Let $\A$ be an online algorithm for a minimization problem. The algorithm $\cal A$ is said to be \emph{$c$-competitive}, if 
$c=\sup_{\beta}\frac{\A_{\beta}}{\OO_{\beta}}$,  where $\A_{\beta}$ and $\OO_{\beta}$ are  the costs of the solution produced by the online algorithm $\A$ and  an optimal offline algorithm, respectively, with respect to an input sequence $\beta$. If $\A$ is a  randomized algorithm, then $\A_{\beta}$ is replaced by the expectation ${\mathbb{E}[\A_{\beta}]}$, and the 
competitive ratio of  $\A$ is  $\sup_{\beta}\frac{\mathbb{E}[\A_{\beta}]}{\OO_{\beta}}$~\cite{BorodinE}.

\subsection{Our Contributions}
{We consider the online hitting set problem when ${\cal X}={\mathbb Z}^d$ and $\cal R$ consists of translated copies of a  geometric object in $\IR^d$.
For lower dimensional objects,  we propose a deterministic online algorithm {$\BPA$}. The general overview of the algorithm is as follows.}

{Depending upon the objects and dimensions, we consider a \emph{filter-set}: a subset $\chi$ of integer points such that 
any input object must contain at least one point of $\chi$.
Our algorithm maintains a hitting set $\A$ consisting of points from $\chi$. Initially $\A=\emptyset$. On receiving a new input  object $\sigma$, if it is not hit by any of the points from $\A$, our online algorithm adds the \emph{best-point} from $\chi$ lying inside $\sigma$ to the set $\A$. For the definition of best-point, we refer to Section~\ref{1.2}.}

{\begin{enumerate}
   \item   When $\cal{R}$ consists of  one-dimensional unit intervals, we have a   $\BPA$ algorithm achieving an optimal competitive ratio of $2$ (Theorem~\ref{thm:int}). 
    \item When $\cal R$ consists of unit balls in $\IR^2$ and $\IR^3$, respectively, we have  $\BPA$ algorithms having  competitive ratios of at most~$4$ and $14$, respectively (Theorem~\ref{2d-balls} and Theorem~\ref{3d-ball}).
        \item When $\cal R$ consists of unit hypercubes in $\IR^2$ and $\IR^3$, respectively, we have  $\BPA$ algorithms having competitive ratios of at most $4$ and $8$, respectively (Theorem~\ref{square_ub} and Theorem~\ref{cube_ub}).
\end{enumerate}}

 {When $\cal R$ consists of unit balls in $\IR^d$,  we propose a deterministic online algorithm {$\NC$} that works
 as follows.
On receiving a new input object $\sigma\subset\mathbb{R}^d$ centered at $c$, 
if it has not been hit by the existing hitting set, then our online algorithm adds the nearest integer point from the center $c$ as the hitting point. If ties happen, our algorithm arbitrarily chooses one of the nearest points as the hitting point. We show that this algorithm achieves a competitive ratio of at most $O(d^4)$ (Theorem~\ref{ball_ub}).}

{ When $\cal R$ consists of unit hypercubes in $\IR^d$, the algorithm $\NC$ achieves an exponential competitive ratio due to the following reasons.
        Let $p$ be a point in offline optimum. Let $\I_p$ be the collection of input hypercubes containing the point $p$. Notice that the center of any hypercube in $\I_p$ lies in  a unit hypercube $H$ centered at $p$. Since $H$ contains exactly $3^d$ integer points, the algorithm $\NC$ might place at most $3^d$ points to hit all the objects in~$\I_p$.}

        {To obtain a better competitive ratio for unit hypercubes in $\IR^d$ ($d\geq 3$), we propose an algorithm, $\textsc{Randomized-}$ $\textsc{Iterative-Reweighting}$, that is similar in nature to an algorithm proposed by Dumitrescu and T{\'{o}}th in~\cite{DumitrescuT22}. Using some  structural properties, we  analyze  this randomized algorithm and show that it has a  competitive ratio of at most~$O(d^2)$ (Theorem~\ref{hyp_ub}).}

{Additionally, we investigated the lower bounds of the hitting set problem for unit balls and unit hypercubes in $\IR^d$, and obtained the following results. When $\cal R$ consists of unit balls in $\IR^d$ ($d<4$) and unit hypercubes in $\IR^d$, we show that every deterministic algorithm has a competitive ratio of at least~$d+1$ (Theorem~\ref{ball_lb} and Theorem~\ref{hyp_lb}).}

 {All the above-mentioned outcomes also hold for the equivalent  geometric set cover problem in the online setup.  A summary of all results obtained in this paper for the online hitting set problem is presented  in Table~\ref{tab:1}.}

 \begin{table}[htbp]
    \centering
    \begin{tabular}{|p{4.5 cm}|p{5.25 cm}|p{5.25 cm}|}
\hline \textbf{Ranges/Objects}  & \textbf{Lower Bound of Competitive Ratio} & \textbf{Upper Bound of Competitive Ratio} \\
\hline
\hline Unit Intervals  & 2 (Theorem~\ref{thm:int})& 2 (Theorem~\ref{thm:int})\\
\hline Unit Disks  & 3 (Theorem~\ref{ball_lb})&4 (Theorem~\ref{2d-balls})\\
\hline Unit Balls in $\IR^3$ & 4 (Theorem~\ref{ball_lb})&14 (Theorem~\ref{3d-ball}).\\
\hline Unit Balls in $\IR^d$ & 4 (Theorem~\ref{ball_lb})& $O(d^4)$ (Theorem~\ref{ball_ub})\\
\hline Unit Squares & 3 (Theorem~\ref{hyp_lb}) & 4 (Theorem~\ref{square_ub})\\
\hline Unit Cubes & 4 (Theorem~\ref{hyp_lb})& 8 (Theorem~\ref{cube_ub})\\
\hline Unit Hypercubes in $\IR^d, d\geq 3$ & $d+1$ (Theorem~\ref{hyp_lb}) & $O(d^2)$ (Theorem~\ref{hyp_ub}) \\
\hline
\end{tabular}
    \caption{{Summary of the results obtained in this paper for the online hitting set problem.}}
    \label{tab:1}
\end{table}

\subsection{Related Work}
The hitting set and set cover problems are classical NP-hard problems~\cite{Karp}.
In the offline setup, if the set $\cal X$ contains points on the real line and $\cal R$ consists of intervals in $\mathbb{R}$, the set cover problem can be solved in polynomial time using a greedy algorithm{~\cite{golumbic}}. However, these problems remain  NP-hard, even when $\cal R$ consists of simple geometric objects like unit disks in $\mathbb{R}^2$~\cite{FOWLER1981133} and ${\cal X}$ is a set of points in ${\IR}^2$. 
Alon et al.~\cite{Alon09} initiated the study of the set cover problem in the online setup. They considered the model where {both} sets $\cal X$ and $\cal R$ are already known, but the order of arrivals of points in $\cal X$ is unknown. Upon the arrival of an uncovered point in $\cal X$, the online algorithm must
choose a range $r\in \cal R$ that covers the point.
The algorithm presented by Alon et al.~\cite{Alon09} has a competitive ratio of $O(\log n \log m)$.
Later, Even and Smorodinsky~\cite{EvenS14} studied the online hitting set problem, where {both} sets $\cal X$ and $\cal R$ are known in advance, but the order of arrival of the input objects in $\cal R$ is unknown. They proposed online algorithms having a competitive ratio of $O(\log n)$ when $\cal R$ consists of half-planes and unit disks in $\IR^2$. They gave matching lower bounds of the competitive ratio for these cases.  They also proposed an online algorithm that achieves an optimal bound of $\Theta(\log n)$ when $\cal R$ consists of intervals in the range $[1,n]$ and  $\cal X$ consists of all  integers in the range $[1,n]$.
In this paper, we 
consider online hitting set problem where  ${\cal X}=\mathbb{Z}^d$ and objects in $\cal R$ consists of unit balls (and hypercubes) in $\IR^d$.
 We consider the model in which $\cal X$ is known in advance, but  objects in $\cal R$ are not known beforehand.

 A variant of the set cover problem  is known as the \emph{unit covering problem} where $\cal X$ is a set of points in $\IR^d$ and the set $\cal R$ consists of all (infinite) possible translated copies of a  given unit object $\sigma^*$ in $\IR^d$. In the online version of the unit covering problem, the set $\cal X$ is not known in advance.
 Charikar et al.~\cite{CharikarCFM04} studied the online version of the unit covering problem where $\sigma^*$ is a unit ball in $\IR^d$.  
 They proposed an online algorithm having a competitive ratio of~$O(2^dd\log d)$. They also proved $\Omega(\log d/\log\log \log d)$ as the lower bound for this problem. Dumitrescu et al.~\cite{DumitrescuGT20} improved both the upper and lower bound of the competitive ratio to $O({1.321}^d)$ and $\Omega(d+1)$, respectively. {In particular, they obtained 5 and 12 competitive ratios, when $\sigma^*$ is a unit ball in $\IR^2$ and $\IR^3$, respectively.} When $\sigma^*$ is a centrally symmetric convex object in $\mathbb{R}^d$,   they proved that the competitive ratio of every deterministic online algorithm is at least $I(\sigma^*)$, where $I(\sigma^*)$ is the illumination number (for definition, see~\cite{DumitrescuGT20}) of the convex object $\sigma^*$. 
 When $\sigma^*\subset \IR^d$ is any object having  aspect$_\infty$ ratio (for definition see~\cite{DeJKS22}) as $\alpha$, 
  a deterministic online algorithm  is known as having a competitive ratio of at most~$\left(\frac{2}{\alpha}\right)^d\left((1+{\alpha})^d-1\right)$ $\log_{(1+\alpha)}(\frac{2}{\alpha}) +1$~\cite{DeJKS22}. Note that the aspect$_\infty$ ratio of any object is in the range $(0,1]$.
Dumitrescu and T{\'{o}}th~\cite{DumitrescuT22} studied another variant of the online unit covering problem where $\cal X$ is a set of points in $\mathbb{Z}^d$. They consider the case when $\sigma^*$  is a hypercube of side length one unit  in $\mathbb{R}^d$.  They~\cite{DumitrescuT22} proved that the competitive ratio of every deterministic online algorithm for this problem is at least $d + 1$.
 They also proposed a randomized online algorithm with a competitive ratio of $O(d^2)$ for this problem. 
For this problem, an equivalent version of the online hitting set problem is as follows:  ${\cal X}=\IR^d$ and the center of the objects in $\cal R$ are from $\mathbb{Z}^d$. To complement their result, in this paper, we consider the online hitting set problem when the ${\cal X}=\mathbb{Z}^d$ and the center of objects in $\cal R$ are from $\IR^d$.

\subsection{Notation and Preliminaries}\label{1.2}
 We use $[n]$ to denote the set $\{1,2,\ldots,n\}$. 
By an \emph{object}, we refer to  a simply connected compact set in $\mathbb{R}^d$ having a nonempty interior. {For any point $p\in\IR^d$, we use $p(x_i)$ to denote the $i$th coordinate of $p$, where $i\in[d]$.}
 An \emph{integer point} is a point $p \in  \mathbb{R}^d$  such that for each $i\in[d]$ the coordinate $p(x_i)$ is an integer. Any two integer points $p$ and $q$ are said to be \emph{consecutive integer points} if there exists an index $j\in[d]$ such that $|p(x_j)-q(x_j)|=1$ and  $p(x_i)=q(x_i)$ for all $i\in[d]\setminus\{j\}$.
  We use $\Q(\sigma)$ to denote the set of integer points contained in an object $\sigma$. For any $\chi\subset\mathbb{Z}^d$, the term $\chi(\sigma)$ denotes the intersection of $\chi$ and $\Q(\sigma)$.
 
The term \emph{integer hypercube} refers to a 
 hypercube $H\subset\mathbb{R}^d$ of side length one having all corners as integer points. {We use $dist(x,y)$ (respectively, ${dist_{\infty}}(x,y)$) to represent the distance between two points $x$ and $y$ under $L_2$-norm (respectively, $L_{\infty}$-norm).} Let $c$ be a point  in $\IR^d$.
We use $H_d(c,r)$ to denote an $L_{\infty}$ ball of radius $r$ centered at $c$. In other words,  $H_d(c,r)=\{x\in\mathbb{R}^d:dist_{\infty}(x,c) \leq r\}$.
 A \emph{unit hypercube}  $H_d(c,1)\subset \mathbb{R}^d$ centered at $c$, is defined as $H_d(c,1)=\{x\in\mathbb{R}^d:dist_{\infty}(x,c) \leq 1\}$. Note that, according to our definition, an integer hypercube is not a unit hypercube.  A \emph{unit ball}  $B_d(c,1)\subset \mathbb{R}^d$ centered at $c$, is defined as $B_d(c,1)=\{x\in\mathbb{R}^d:dist(x,c)\leq 1\}$.
 Throughout the paper, if not stated otherwise, the term {hypercube} is used to refer to an axis-aligned unit hypercube and the term {ball} is used to refer to a unit ball.

 {Let us define  a `relation' $\prec$ among distinct points in $\IR^d$ as follows.
Note that for any pair of distinct points $p$ and $q$ in $\IR^d$, there exists a unique index $i\in [d]$ such that  $p(x_{i})\neq q(x_{i})$ and  $p(x_j)=q(x_j)$ for each $j\in \{i+1,\ldots,d\}$.
 If $p(x_{i})< q(x_{i})$,  we say that $p \prec q$; otherwise  $q \prec p$. Note that this gives a strict total ordering for any set $P\subset \IR^d$ of distinct elements.
 For a set $P$ of distinct points, a point $p^*\in P$ is defined as the \emph{best-point}   if $q \prec p^{*}$, for all $q(\neq p^{*})\in P$.}

\subsection{Organization}
In Section~\ref{sec:interval}, we present the lower and upper bound of the competitive ratio for hitting one-dimensional intervals.
Next, in Section~\ref{sec:d-ball}, for hitting unit balls in $\IR^d$, we give the lower and upper bound of the competitive ratio. 
Section~\ref{sec:d-hyp} consists of the lower and upper bound of the competitive ratio for hitting axis-aligned unit hypercubes in $\IR^d$.
Later, in Section~\ref{sec:cover}, we summarize the results obtained for the unit covering problem.
 Eventually, in Section~\ref{Conclusion}, we conclude.


\section{Hitting Set Problem for Unit Intervals}\label{sec:interval}

 {We first} consider when objects are one-dimensional unit hypercubes, i.e., unit intervals.

\begin{theorem}\label{thm:int}
For hitting unit intervals using points from $\mathbb{Z}$, there exists a deterministic online algorithm that achieves a competitive ratio at most $2$. This result is tight: the competitive ratio of any
deterministic online algorithm for this problem is at least $2$.
\end{theorem}

\begin{proof}
We first prove the upper bound of the competitive ratio.
Let $\Lambda=\{q{\bf e}_1\ |\ q\in \mathbb{Z}\}$ be the integer lattice generated by standard unit vector ${\bf e}_1$. Partition the whole integer lattice using integer point from $\chi=\{2q {\bf e_1}\ |\ q\in \mathbb{Z}\}$. Note that any unit interval can contain at least one and at most two integer points from $\chi$.
Our algorithm maintains a hitting set $\A$. Initially $\A=\emptyset$.
On receiving a new input interval $\sigma$, if it is not hit by any of the points from $\A$, our algorithm adds one integer point from $\chi$ contained in the interval $\sigma$ to the set $\A$.
 
Let $\I$ be the set of input intervals  presented to the algorithm.
Let $\OO$ be an offline optimal hitting set for $\I$.  
 Let $\A'=\A\setminus \{\A\cap \OO\}$ and  $\OO'=\OO\setminus \{\A\cap \OO\}$. 
 Let $p \in \OO'$ be an integer point and let $\I_p\subseteq \I$ be the set of input intervals that are hit by the point $p$.
 Let $\A_{p}\subseteq \A'$ be the set of points used by our algorithm to hit the intervals in $\I_p$.   If $p\in\chi$, then $\A_{p}$  contains either $\{p,p+2\}$ or $\{p-2,p\}$ from $\chi$, since $p\notin \A\cap \OO$, {we have} $|\A_p|\leq 1$; otherwise, $\A_{p}$  contains at most two integer points: $p-1$ and $p+1$ from $\chi$. 
Therefore, $|\A_p|\leq 2$. Since $\A'=\cup_{p\in \OO'}\A_p$, we have $|\A'|\leq\sum_{p\in \OO'}|\A_p|\leq 2\times |\OO'|$. Note that $\frac{|\A'| }{|\OO'|}\leq 2$ implies $\frac{|\A|}{|\OO|}\leq 2$.
Thus, the competitive ratio of our algorithm is at most~2.

To prove the lower bound of the competitive ratio, we construct a sequence of intervals  $\sigma_1,\sigma_2$ adaptively such that any online algorithm needs to place two integer points; while an offline optimum needs just one point.
Initially, we present a unit interval $\sigma_1=[x,x+2]$, where $x\in \mathbb{Z}$. Any online algorithm places an integer point $h_1=x+i$, where $i\in\{0,1,2\}$,  to hit the interval $\sigma_1$. For any choice of $i\in\{0,1,2\}$ for the hitting point $h_1$, it is always possible to present another interval $\sigma_2$ that does not contain the point $h_1=x+i$ but contains the  point $x'=x+((i+1)\mod 3)\in\{x,x+1,x+2\}$. Hence, the theorem follows. 
\end{proof}


\section{Hitting Set Problem for Unit Balls}\label{sec:d-ball}
{In this section, we present $\BPA$ algorithms for unit balls in $\IR^2$ and $\IR^3$. After that, we present the analysis of the algorithm $\NC$ for unit balls in $\IR^d$. Finally, we give a lower bound for hitting unit balls in $\IR^d$ ($d<4$).}

\subsection{Unit Balls in $\IR^2$ and $\IR^3$}
Let $\Lambda_d=\{\alpha_1 {\bf e}_1+\alpha_2 {\bf e}_2+\ldots+\alpha_d {\bf e}_d\ |\ \alpha_i\in \mathbb{Z},\ \forall\ i\in[d]\}$   be the integer lattice in $\IR^d$ generated by standard unit vectors ${\bf e}_1$, ${\bf e}_2,\ldots,{\bf e}_d$. Consider a subset $\chi_d\subset \Lambda_d$ defined as follows:\\ $\chi_d=\{\alpha_1 {\bf u}_1+\alpha_2 {\bf u}_2+\ldots+\alpha_d {\bf u}_d\ |\ \alpha_i\in \mathbb{Z},\ \forall\ i\in[d]\}$. Here, for $d\leq 4$ we have
\[
 {\bf u}_i=
\begin{cases}
 2{\bf e}_1, & \text{for } i=1\\
   {\bf e}_{i-1}+{\bf e}_i,& \text{for } i\in[d]\setminus\{1\}.
  
\end{cases}
\]

\begin{lemma}\label{lem:correct}
    For $d\leq 4$, each unit ball  $B_d(r,1)$ centered at any point $r \in \IR^d$ contains at least one point of $\chi_d$.
\end{lemma}
\begin{proof}
As per the definition of $\chi_d$, precisely one among every two consecutive integer points belongs to the set $\chi_d$. To prove the lemma, it is sufficient to prove that a unit ball $B_d(r,1)\subset \IR^d$ centered at any point $r\in \IR^d$ contains at least  two consecutive integer points. Note that any real number $x\in\IR$ can be expressed as $x=y+z$, where $z\in\mathbb{Z}$ and $y\in\left(\frac{-1}{2},\frac{1}{2}\right]$.
Let $r=(z(x_1)+y(x_1),z(x_2)+y(x_2),\ldots,z(x_d)+y(x_d))$ and $z=(z(x_1),z(x_2),\ldots,z(x_d))$. To show that a point $p\in\IR^d$ belongs to the unit ball $B_d(r,1)$, we need to show that the $dist(r,p)\leq 1$. Now, consider the square of the distance between $r$ and $z$ as follows

\begin{align*}
dist^2(r,z)=&\sum_{i=1}^d\left(z(x_i)-r(x_i)\right)^2\\
=&(z(x_1)-(z(x_1)+y(x_1)))^2+(z(x_2)-(z(x_2)+y(x_2)))^2+\ldots+(z(x_d)-(z(x_d)+y(x_d)))^2\\
=& y(x_1)^2+y(x_2)^2+\ldots+y(x_d)^2\\
 \leq&\ d\left(\frac{1}{2}\right)^2\\
 \leq&1.
\end{align*}  

The last inequality follows because $d\leq 4$.
Let $t\in[d]$ be an index such that $|y(x_t)|=\max\{|y(x_i)|:i\in[d]\text{ and } d\leq4 \}$.
       Let $z' \in \mathbb{Z}^d$ be an integer point such that
      \[
      z'(x_i)=
      \begin{cases}
     z(x_i)+1,& \text{ if $i=t$} \\
      z(x_i),& \text{ otherwise}. 
      \end{cases}
      \]
      Now, consider the square of the distance between $z'$ and $r$ as follows

\begin{align*}
dist^2(r,z')=&\sum_{i=1}^d\left(z'(x_i)-r(x_i)\right)^2\\
=&((z(x_t)+1)-(z(x_t)+y(x_t)))^2\ +\sum_{i\in[d]\setminus\{t\}}\left(z(x_i)-(z(x_i)+y(x_i))\right)^2\\
=&\ 1-2y(x_t)+y(x_1)^2+y(x_2)^2+\ldots+y(x_d)^2\\
\leq&\ 1-2y(x_t)+2y(x_t)=1.
\end{align*} 

 Here, the last inequality follows due to the following: since $|y(x_t)|=\max\{|y(x_i)|:i\in[d]\}$, $d\in[4]$  and $y(x_t)\leq\frac{1}{2}$, we have $y(x_1)^2+y(x_2)^2+\ldots+y(x_d)^2\leq\ d \left(y(x_t)^2\right)\leq \frac{d}{2}y(x_t) \leq 2y(x_t)$. 
Note that the distance of $r$ from both integer points $z$ and $z'$ is less than or equal to 1. Since $z$ and $z'$ are consecutive integer points, one of them must belong to $\chi_d$ and the ball $B_d(r,1)$ contains at least one integer point of $\chi_d$.
\end{proof}
\begin{theorem}\label{3d-ball}
    For hitting unit balls using points in $\mathbb{Z}^3$, there exists a deterministic online algorithm that achieves a competitive ratio of at most~14.
\end{theorem}

\begin{figure}[htbp]
	\centering
 \hfill
\begin{subfigure}[b]{0.32\textwidth}
	\centering
	\includegraphics[width=23.5 mm]{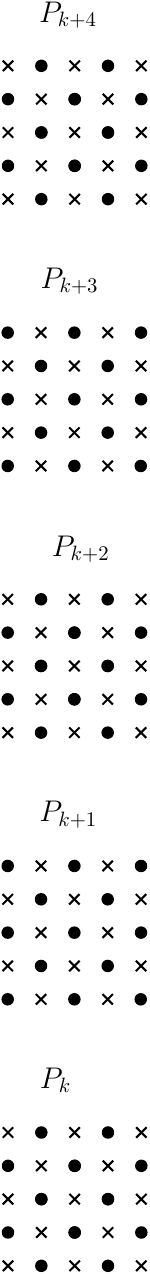}
	\caption{}
	\label{fig:bp}
\end{subfigure}
\hfill
\begin{subfigure}[b]{0.32\textwidth}
	\centering
	\includegraphics[width=30 mm]{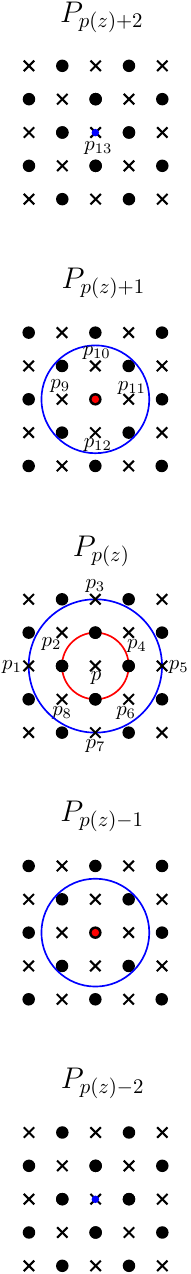}
	\caption{}
	\label{fig:ib_(-2)}
\end{subfigure}
\begin{subfigure}[b]{0.32\textwidth}
	\centering
	\includegraphics[width=23.7 mm]{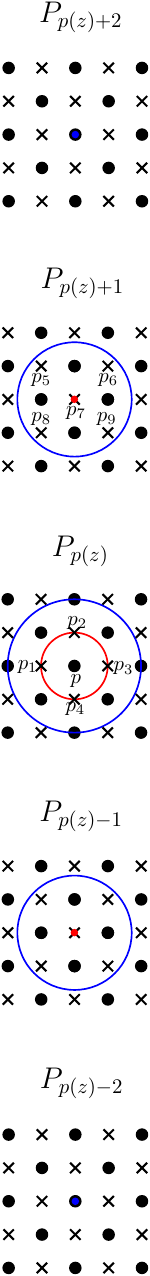}
	\caption{}
	\label{fig:nb_(-2)}
\end{subfigure}
 \caption{(a). {The projections of  planes  $P_{k}, P_{k+1},  \ldots, P_{k+4}$ over a rectangular region.} Illustration of Theorem~\ref{3d-ball}. Here, the boundary of balls $B_3(p,1)$ and $B_3(p,2)$ are represented with red and blue color, respectively, (b) Case~1.1 and (c) Case~1.2.}
   \label{fig:ball}
\end{figure}

\begin{proof}
For the sake of simplicity, throughout the proof, we use $\chi$ to represent $\chi_3$.
For any $k\in  \mathbb{Z}$, we use $P_k$  to denote the plane parallel to $xy$-plane with $z$-coordinate value $k$. The projections of  planes  $P_{k}, P_{k+1},  \ldots, P_{k+4}$  over a rectangular region are depicted in {Fig.~\ref{fig:bp}}. 
Observe that $P_{k}\cap\chi$ and $P_{k+1}\cap\chi$ are translated copies of each other by 1 unit of $y$-coordinate.
Algorithm $\BPA$ maintains a hitting set $\A$ consisting of points of $\chi$. On receiving a new input unit ball $\sigma$, if it is not hit by any of the points from $\A$ then the algorithm adds the best-point of $\chi$ lying inside $\sigma$ to the set $\A$.
{Correctness of the algorithm follows from Lemma~\ref{lem:correct}}.

Let $\I$ be the set of input balls presented to the algorithm.
Let $\OO$ be an offline optimal hitting set for $\I$.  
 Let {$\A'=\A\setminus \{\A\cap \OO\}$ and  $\OO'=\OO\setminus \{\A\cap \OO\}$}. 
Let $p \in \OO'$, and let $\I_p\subseteq \I$ be the set of input balls containing the point $p$.
 Let $\A_{p}\subseteq \A'$ be the set of hitting points placed by our algorithm to hit explicitly when some ball in $\I_p$ arrives.
  In the following lemma, we prove that the cardinality of
$\A_p$ is bounded by 14. 
Since $\A'=\cup_{p\in \OO'}\A_p$, we have $|\A'|\leq\sum_{p\in \OO'}|\A_p|\leq 14\times |\OO'|$. Note that $\frac{|\A'| }{|\OO'|}\leq 14$ implies $\frac{|\A|}{|\OO|}\leq 14$.
Thus, the competitive ratio of our algorithm is at most~14.
\end{proof}
 \begin{lemma}\label{ball_3d}
 $|\A_p|\leq 14$.
 \end{lemma}
 \begin{proof}
 Observe that the center of each $\sigma \in \I_p$ lies in the region $B_3(p,1)$, and to hit balls of $\I_p$, our algorithm places integer points from $\chi(B_3(p,2))$.
 Therefore, $\A_p$ contains points from $\chi(B_3(p,2))$.
Let $p(z)$ be the z-coordinate value of the point $p$.
 Note that the ball $B_3(p,2)$ contains integer points only from five planes, 
 namely, $P_{p(z)+2},P_{p(z)+1},P_{p(z)},P_{p(z)-1}$ and $P_{p(z)-2}$. 
 As per the definition of $\chi$, we know that $P_{p(z)+2}\cap\chi$,  $P_{p(z)}\cap\chi$ and $P_{p(z)-2}\cap\chi$ are same. Similarly,  $P_{p(z)+1}\cap\chi$ is same as $P_{p(z)-1}\cap\chi$. 
Observe that, if the center of the unit ball $B_3(p,1)$ coincides with some point of $\chi$, then the ball contains only one point of $\chi$ (see planes $P_{p(z)-1}, P_{p(z)}$ and $P_{p(z)+1}$ in Fig.~\ref{fig:ib_(-2)}); otherwise, it contains six points of $\chi$ (see planes $P_{p(z)-1}, P_{p(z)}$ and $P_{p(z)+1}$ in Fig.~\ref{fig:nb_(-2)}). As a result, we have the following two cases.

\noindent
\textbf{Case 1:} 
$|\chi(B_3(p,1))|=1$. In this case, $p\in\chi$.
Representative figures of five planes $P_{p(z)-2}$, $P_{p(z)-1}$, $P_{p(z)}$, $P_{p(z)+1}$ and $P_{p(z)+2}$ intersecting the ball $B_3(p,2)$ are  shown in Fig.~\ref{fig:ib_(-2)}.
Observe that  $B_3(p,2)$ contains 19 integer points of $\chi$ including $p$. As per the definition of $\A_p$, we know $p\notin \A_p$.
 Now, we show that none of the points from plane $P_{p(z)-1}$ and $P_{p(z)-2}$ are in $\A_p$. Here,  we want to remind the reader  that any unit ball $\sigma\in\I_p$  contains  the point $p$. For any point $p'\in\chi\cap P_{p(z)-2}$ and $p''\in\chi\cap P_{p(z)-1}$, we have $p',p''\prec p$, thus our algorithm does not add any point from $\chi\cap P_{p(z)-1}$ and $\chi\cap P_{p(z)-2}$ to $\A_p$. Hence, $|\A_p|\leq 13$.
\\
\textbf{Case 2:}  $|\chi(B_3(p,1))|=6$. Notice that, in this case, $p\notin\chi$. Representative figures of five planes $P_{p(z)-2}$, $P_{p(z)-1}$, $P_{p(z)}$, $P_{p(z)+1}$ and $P_{p(z)+2}$ intersecting the ball $B_3(p,2)$ are  shown in Fig.~\ref{fig:nb_(-2)}.
Observe that the ball $B_3(p,2)$ contains only fourteen integer points of $\chi$. Hence, $|\A_p|\leq 14$.
\end{proof}

\begin{theorem}\label{2d-balls}
    For hitting unit disks using points in $\mathbb{Z}^2$, there exists a deterministic online algorithm that  achieves a competitive ratio of at most~4.
\end{theorem}
\begin{proof}

The proof is similar to Theorem~\ref{3d-ball}. Here,
the integer lattice $\Lambda=\{q {\bf e}_1+r {\bf e}_2 | q,r\in \mathbb{Z}\}$ in $\IR^2$ is  generated by standard unit vectors ${\bf e}_1$ and ${\bf e}_2$ and the $\chi\subset \Lambda$ is defined as $\chi=\{q {\bf u}+r {\bf v} |  q,r\in \mathbb{Z}\}$, where ${\bf u}=2{\bf e}_1$ and ${\bf v}={\bf e}_1+{\bf e}_2$. Here, we prove that $|\A_p|\leq 4$. Thus, the competitive ratio of algorithm $\BPA$ is at most 4.
 Observe that, if the center of the ball $B_3(p,1)$ coincides with some point in $\chi$, then $B_3(p,1)$ contains only 1 point of $\chi$ (see plane $P_{p(z)}$ in Fig.~\ref{fig:ib_(-2)}); otherwise, it contains 4 points of $\chi$ (see plane $P_{p(z)}$ in Fig.~\ref{fig:nb_(-2)}). Similar to Lemma~\ref{ball_3d}, we have two cases.\\
\textbf{Case 1:} $p\in\chi$. Note that $B_2(p,2)$ contains 9 integer points of $\chi$ (see plane $P_{p(z)}$ in Fig.\ref{fig:ib_(-2)}). As per the definition of $\A_p$, we know $p\notin \A_p$. Let us consider a unit disk $\sigma\in\I_p$ that contains the point $p_1$. Since $p_1\prec p$, our algorithm does not add $p_1$ to $A_p$ upon the arrival of $\sigma$. In a similar way, one can observe that none of the points $\{p_2,p_6,p_7,p_8\}$ are in $\A_p$. As a result $|\A_p|\leq 3$.\\
\textbf{Case 2:} $p\notin \chi$. Observe that $B_2(p,2)$ contains only four integer points of $\chi$ (see plane $P_{p(z)}$ in Fig.\ref{fig:nb_(-2)}). Hence, $|\A_p|\leq4.$ \end{proof}

\subsection{Unit Balls in $\mathbb{R}^d$}\label{sec:balls}
In this subsection, we present the upper bound on the competitive ratio for hitting unit balls in $\IR^d$.

\begin{theorem}\label{ball_ub}
For hitting unit balls using points in $\mathbb{Z}^d$, {the algorithm $\NC$}  achieves a  competitive ratio of at most~$O(d^4)$, when $d\in\mathbb{N}$.
\end{theorem}

\begin{proof}
Let $\A$ and $\OO$ be the hitting set returned by our online algorithm and an offline optimal, respectively.  Let $p \in \OO$ be any point.
Note that a unit ball $B_d(p,1)$ centered at $p$ contains all the centers of unit balls that can be hit by the point $p$. For simplicity, throughout the proof, let us assume that the point $p$ coincides with the origin. Let $\A_p$ be the set of hitting points placed by our online algorithm to pierce the ball having a center in $B_d(p,1)$.
It is easy to see that $\A=\cup_{p\in \OO}\A_p$.
Therefore, the competitive ratio of our algorithm is  upper bounded by  $\max_{p\in\OO}|\A_p|$. For any point, $r\in B_d(p,1)$, the maximum distance from $r$ to its nearest integer point can be at most one (the maximum distance from the center $c$ of the unit ball to any point $r\in B_d(c,1)$ is at most one). Therefore, a ball $B_d(p,2)$ centered  at $p$ having radius $2$ will contain all nearest integer points for all the centers $r$ lying in the ball $B_d(p,1)$. To complete the proof, we only need to calculate the cardinality of the set $\{z\in \mathbb{Z}^d: \sum_{i=1}^d |z|^2\leq 4\}$. In other words, we need to count the number of $z=(z(x_1),z(x_2),\ldots,z(x_d))\in\mathbb{Z}^d$ that satisfies:
\begin{equation}\label{eqn_z}
z(x_1)^2+z(x_2)^2+\ldots+z(x_d)^2\leq 4.
\end{equation}
\noindent
Note that to satisfy Equation~(\ref{eqn_z}), the coordinates of $z$ cannot be other than $\{-2,-1,0,1,2\}$
\begin{itemize}
    \item When all $d$ coordinates are $0$. There is only one possibility for this.
    \item When exactly one coordinate is nonzero. There will be ${d \choose 1}$ many choices for the position of the nonzero coordinate. Now, observe that for each nonzero coordinate, we have four choices $\{-2,-1,1,2\}$. 
   So, for this case, there will be a total of $4d$ integer points satisfying Equation~(\ref{eqn_z}).
   \item  Note that any integer point having more than four nonzero coordinates will not satisfy Equation~(\ref{eqn_z}). Now consider exactly $i$ nonzero coordinates for $i=2,3,4$. 
    There will be ${d \choose i}$ many choices for the position of the nonzero coordinates. Now observe that if any of the nonzero coordinates is $\{-2,2\}$, then the integer point will not satisfy Equation~(\ref{eqn_z}). Therefore, for each nonzero coordinate, we have just two choices $\{-1,1\}$. Thus, there will be a total of $2^i{d \choose i}$ integer points satisfying Equation~(\ref{eqn_z}).

 \end{itemize}
Now, from the above cases, there will be at most $1+4d+\sum_{i=2}^4 2^i{d \choose i}=O(d^4)$ integer points satisfying Equation~(\ref{eqn_z}). Hence, we have $|\A|\leq O(d^4)|\OO|$.
\end{proof}

\subsection{Lower Bound for $d<4$}
To obtain a lower bound of the competitive ratio, we think of a game between two players: Alice and Bob. Here, Alice plays the role of the adversary, and Bob plays the role of the online algorithm. In each round of the game, Alice presents a unit ball such that Bob needs to place a new hitting point.
We show that Alice can present an input sequence of balls  $\sigma_1,\sigma_2,\ldots,\sigma_{d+1}\subset \mathbb{R}^d$, centered at $c_1,c_2,\ldots,c_{d+1}$, respectively, depending on the position of hitting points placed by Bob, for which Bob needs to place $d+1$ integer points; while the offline optimum needs just one point {(for illustration in two-dimensions, see Figure~\ref{fig:equi_class})}.
For the sake of simplicity, let us assume that the center $c_1$ of the first ball $\sigma_1$ coincides with the origin. Note that the ball $\sigma_1$ contains exactly $2d$ integer points $\P=\{p_1,p_2,\ldots,p_{2d}\}$ apart from the origin. The coordinates of these points are given below:

\begin{equation}\label{eqn_p}
p_k(x_j) =
    \begin{cases}
     \ \ 1,\quad \text{ if $k=j$,} & \text{for $k,j\in[d]$}\\
      -1,\quad \text{ if $k=d+j$,} & \text{for $k\in[2d]\setminus[d]$ \& $j\in[d]$}\\ 
     \ \ 0,\quad \text{ otherwise}.&
    \end{cases}      
\end{equation}

Let $\P_1=\{p_{1},p_{2},\ldots,p_{d}\}$ and  $\P_2=\{p_{d+1},p_{d+2},\ldots,p_{2d}\}$. To hit the input ball $\sigma_1$, Bob needs to choose a point $h_1\in \P_1\cup\P_2\cup \{c_1\}$. 
Depending on the position of $h_1$, Alice presents a ball $\sigma_2$ centered at a point $c_2$ that satisfies the following:

\begin{equation}\label{eq:c_2}
c_2 =
    \begin{cases}
    \left(\frac{1}{2}+\epsilon_d,\frac{1}{2}+\epsilon_d,\ldots,\frac{1}{2}+\epsilon_d\right),& \text{if $h_1\in \P_2\cup\{c_1\}$}\\
      \left(-(\frac{1}{2}+\epsilon_d),-(\frac{1}{2}+\epsilon_d),\ldots,-(\frac{1}{2}+\epsilon_d)\right), & \text{otherwise (i.e., if $h_1\in \P_1$)},
    \end{cases}      
\end{equation}

where 
 the value of $\epsilon_d$ is $0.5$ and $0.15$ for $d=2\text{ and } 3$, respectively.
Note that $\sigma_2$ does not contain the point $c_1$.

 \begin{figure}[htbp]
    \centering
    \includegraphics[width=50 mm]{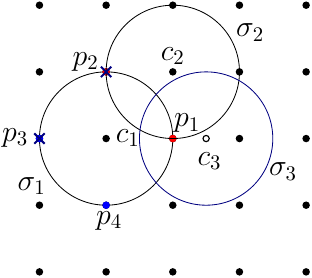}
    \caption{{Illustration of the lower bound for unit balls in $\IR^2$. Here, $\P_1=\{p_{1},p_{2}\}$ and  $\P_2=\{p_{3},p_4\}$.  Let $\sigma_1$ be the first ball presented by  Alice. To hit the input ball $\sigma_1$, Bob chooses a point $h_1=p_3\in \P_2$. 
 Alice presents the next ball $\sigma_2$ centered at $c_2=
    \left(1,1\right)$ such that $\sigma_2$ contains all the points of $\P_1$ but does not contain any point from $\P_2\cup c_1$.  To hit $\sigma_2$, Bob chooses a point $h_2=p_2$. Alice presents $\sigma_3$ centered at a point $c_3=
    \left(1,0\right)$. The ball $\sigma_3$ does not contain $h_1$ and $h_2$. To hit $\sigma_3$, Bob place $h_3=p_1$. To hit $\sigma_1,\sigma_2$ and $\sigma_3$, any offline optimum will place $p_1$ as the hitting point.}}
    \label{fig:equi_class}
\end{figure}

\begin{lemma}\label{claim_1}
\begin{itemize}
 \item[(i)]  If $h_1\in \P_2\cup\{c_1\}$, then ${\cal Q}(\sigma_2)$ contains all the points of $\P_1$ and it does not contain any point of~$\P_2\cup \{c_1\}$.
\item[(ii)]  If $h_1\in \P_1$, then ${\cal Q}(\sigma_2)$ contains all the points of $\P_2$ and it does not contain any point of~$\P_1$
\end{itemize}
\end{lemma}

\begin{proof}
We prove part(i) of the lemma statement. The proof of part(ii) would be similar in nature.    Assume that  $h_1\in\P_2\cup\{c_1\}$. According to~(\ref{eq:c_2}), Alice presents $\sigma_2$ centered at $c_2=\left(\frac{1}{2}+\epsilon_d,\frac{1}{2}+\epsilon_d,\ldots,\frac{1}{2}+\epsilon_d\right)$.
Note that  the ball $\sigma_2$ does not contain the point $c_1$.
To see that $\sigma_2$ does not contain any point from $\P_2$, observe that 
for each $p_k\in\P_2$, we have  
\begin{align*}
    dist(c_2,p_k)^2=&\left(\frac{3}{2}+\epsilon_d\right)^2+\sum_{j\in[d ]\setminus\{k-d\}} \left(\frac{1}{2}+\epsilon_d\right)^2>1.
\end{align*}
Finally, we prove that the ball $\sigma_2$ contains all the points of $\P_1$.
For each $p_k\in\P_1$, we have\[
dist(c_2,p_k)^2=\left(\frac{1}{2}-\epsilon_d\right)^2+\sum_{j\in[d]\setminus \{k\}} \left(\frac{1}{2}+\epsilon_d\right)^2=\left(-\frac{1}{2}+\epsilon_d\right)^2+(d-1)\left(\frac{1}{2}+\epsilon_d\right)^2\leq1,
\]
The last inequality follows by placing the specific values of $\epsilon_d$, i.e., 0.5 and 0.15 for $d=2$ and $3$, respectively.
\noindent
Hence, the lemma follows. \end{proof}

From now onwards, we assume that Bob chooses $h_1\in\P_2\cup\{c_1\}$. The other case is similar in nature. Now, we show by induction that Alice and Bob can play the game for  the next $d+1$ rounds maintaining the following two invariants:  For $i=2,\ldots, {d+1}$ when Alice presents  balls  $\sigma_2,\ldots, \sigma_i$ and Bob presents piercing points $p_{\pi(2)}, p_{\pi(3)},\ldots, p_{\pi(i-1)}\in\P_1$. 
\begin{itemize}
    \item[(I)] The ball $\sigma_i\subset \mathbb{R}^d$ does not contain any previously placed hitting point $h_j\in \mathbb{Z}^d$, for 
$j\in[i-1]$.
   \item[(II)] The ball $\sigma_{i}$ contains all the points from $\P_1\setminus\{p_{\pi(2)}, p_{\pi(3)}\ldots, p_{\pi(i-1)}\}$. 
  \end{itemize}
Invariant (I) ensures that Bob needs a new point to hit $\sigma_i$. On the other hand, Invariant (II) ensures that $\cap\sigma_i$ contains a point from $\P_1$ that is not used by Bob.
 For $i=2$, due to Lemma~\ref{claim_1}, both the invariants are maintained. 
At the beginning of the round $i$ (for $i=2,\ldots,d$), assume that both invariants hold.
Let $\Pi=\{\pi(2),\pi(3),\ldots,\pi(i)\}$ be the set of indices of integer points chosen from $\P_1$ to hit the previously arrived balls.
Depending on the position of the hitting point $p_{\pi(i)}$, Alice presents a ball $\sigma_{i+1}$, in the $(i+1)$th round of the game, centering at $c_{i+1}$ that satisfies the following.

\begin{equation}\label{eqn}
c_{i+1}(x_j)=
\begin{cases}
     \left(\frac{3}{2}\right)^{(i-1)}c_2(x_j),& \text{for all $j\in[d]\setminus \Pi$, and}\\
      0, & \text{for $j\in \Pi$}.
    \end{cases}      
\end{equation}

\noindent
${\bullet}$ First, we prove that $\sigma_{i+1}$ does not contain the first hitting point $h_1$. Observe that $dist(c_{i+1},h_1)^2=\sum_{j\in [d]} \left(c_{i+1}(x_j)-h_1(x_j)\right)^2$.
 Note that for $j\in \Pi$, the value of $c_{i+1}(x_j)$ is zero. So we have 
 \begin{align*}
 dist(c_{i+1},h_1)^2=&\sum_{j\in \Pi} \left(0-h_1(x_j)\right)^2+\sum_{j\in[d]\setminus \Pi} \left(\left(\frac{3}{2}\right)^{(i-1)}\left(\frac{1}{2}+\epsilon_d\right)-h_1(x_j)\right)^2.
 \end{align*}
 
 If $h_1=c_1$, then we have $dist(c_{i+1},h_1)^2= \quad0 + (d-i+1)\left(\frac{3}{2}\right)^{2(i-1)}\left(\frac{1}{2}+\epsilon_d\right)^2>1$.
 If $h_1=p_k\in\P_2$, then we have the following two sub-cases. 
If $(k-d)\in \Pi$, we have $dist(c_{i+1},h_1)^2= 1 + (d-i+1)\left(\frac{3}{2}\right)^{2(i-1)}\left(\frac{1}{2}+\epsilon_d\right)^2>1$,
 otherwise (i.e., $(k-d)\in[d]\setminus \Pi$), we have $dist(c_{i+1},h_1)^2= (d-i+1)\left(\left(\frac{3}{2}\right)^{(i-1)}\left(\frac{1}{2}+\epsilon_d\right)+1\right)^2>1$.
Now, we show that  $\sigma_{i+1}$ does not contain any of the previously placed hitting points of $\P_1$. Here, for any $p_{\pi(k)}\in\{p_{\pi(2)}, p_{\pi(3)}\ldots, p_{\pi(i)}\}$, we have
 
\begin{align*}
dist(c_{i+1},p_{\pi(k)})^2=&\sum_{j\in [d]} \left(c_{i+1}(x_j)-p_{\pi(k)}(x_j)\right)^2.
\end{align*}

\noindent
Note that for $j\in \Pi$, $c_{i+1}(x_j)=0$, and $p_{\pi(k)} \in\P_1$ has only one nonzero coordinate that is the $\pi(k)$th coordinate with value 1 and $\pi(k)\in\Pi$. Therefore, we have

\begin{align*}
dist(c_{i+1},p_{\pi(k)})^2=&\sum_{j\in \Pi} \left(0-p_{\pi(k)}(x_j)\right)^2+\sum_{j\in[d]\setminus \Pi} \left(\left(\frac{3}{2}\right)^{(i-1)}\left(\frac{1}{2}+\epsilon_d\right)\right)^2\\
=& 1+(d-i+1)\left(\frac{3}{2}\right)^{2(i-1)}\left(\frac{1}{2}+\epsilon_d\right)^2>1. 
\end{align*}

Therefore, the distance between the center $c_{i+1}$ and previously placed hitting points $\{p_{\pi(2)}, p_{\pi(3)}\ldots, p_{\pi(i)}\}$ is greater than one.
Hence, invariant (I) holds.\\

\noindent
${\bullet}$ Now, we show that $\sigma_{i+1}$ contains all $(d-i+1)$ integer points from $\P_1\setminus\{p_{\pi(2)}, p_{\pi(3)}\ldots, p_{\pi(i)}\}$.  Here, for any $p_k\in\P_1\setminus\{p_{\pi(2)}, p_{\pi(3)}\ldots, p_{\pi(i)}\}$, we have
 
    \begin{align*}
        dist(c_{i+1},p_k)^2=&\sum_{j\in [d]} \left(c_{i+1}(x_j)-p_k(x_j)\right)^2\\
        =&\sum_{j\in \Pi} \left(c_{i+1}(x_j)-p_k(x_j)\right)^2+\sum_{j\in[d]\setminus \Pi} \left(\left(\frac{3}{2}\right)^{(i-1)}\left(\frac{1}{2}+\epsilon_d\right)-p_k(x_j)\right)^2.
        \end{align*}
        
Note that for $j \in  \Pi$, both $c_{i+1}(x_j)$ and $p_k(x_j)$ are zero. Here, $p_k$ has only one nonzero coordinate, which is the $k$th coordinate with value one and $k\notin\Pi$. Therefore, we have
       
        \begin{align}\label{eqn:why}
        \begin{split}
      dist(c_{i+1},p_k)^2=&\ 0
       +\left(\left(\frac{3}{2}\right)^{(i-1)}\left(\frac{1}{2}+\epsilon_d\right)-1\right)^2+\sum_{j\in [d]\setminus \{\Pi\cup\{k\}\}}\left(\left(\frac{3}{2}\right)^{(i-1)}\left(\frac{1}{2}+\epsilon_d\right)\right)^2\\
      =& \left(\left(\frac{3}{2}\right)^{(i-1)}\left(\frac{1}{2}+\epsilon_d\right)-1\right)^2+(d-i) \left(\left(\frac{3}{2}\right)^{(i-1)}\left(\frac{1}{2}+\epsilon_d\right)\right)^2\leq1.
      \end{split}
    \end{align} 
    
    The last inequality follows by placing specific values of $\epsilon_d$, i.e., $0.5$ and $0.15$ for $d=2$ and $3$, respectively.
   Hence, invariant (II) is maintained.

As a result, any online algorithm needs $d+1$ hitting points $\{p_{\pi(2)},p_{\pi(3)}\ldots p_{\pi(d+1)}\}$ and $h_1$; whereas an offline optimum needs just one point $p_{\pi(d+1)}$.
Thus, we have the following theorem.
\begin{theorem}\label{ball_lb}
The competitive ratio of every deterministic online algorithm is at least~$d+1$ for hitting unit balls in $\mathbb{R}^d$ using points in $\mathbb{Z}^d$, where  $d<4$.
\end{theorem}

\noindent
\textbf{Remark~1:} In equation~\eqref{eqn:why}, for any $\epsilon_d>0$ and $d\geq 4$, {the value of $dist(c_{i+1},p_k)^2$ is strictly greater than 1. As a result, invariant (II) is not satisfied. Thus, the proof is only valid for $d< 4$.}


\section{Hitting Set Problem for Unit Hypercubes}\label{sec:d-hyp}
{In this section, we start by presenting $\BPA$ algorithms for unit hypercubes in $\IR^2$ and $\IR^3$. After that, we present some structural properties of hypercubes in $\IR^d$ that will play a crucial role in the analysis of the algorithm, $\RIR$,  for unit hypercubes in $\IR^d$ ($d\geq 3$). Finally, we give a lower bound for the problem.}

\subsection{Unit Hypercubes in $\IR^2$ and $\IR^3$}
Let $\Lambda_d=\{\alpha_1 {\bf e}_1+\alpha_2 {\bf e}_2+\ldots+\alpha_d {\bf e}_d\ |\ \alpha_i\in \mathbb{Z},\ \forall i\in[d]\} $  be the integer lattice in $\IR^d$ generated by standard unit vectors ${\bf e}_1$, ${\bf e}_2,\ldots,{\bf e}_d$. Consider a subset $\chi_d\subset \Lambda_d$ defined as follows:\\
$\chi_d=\{\alpha_1 {\bf u}_1+\alpha_2 {\bf u}_2+\ldots+\alpha_d {\bf u}_d\ |\ \alpha_i\in \mathbb{Z},\ \forall\ i\in[d]\}$.
Here, we have
\[
 {\bf u}_i=
\begin{cases}
 2{\bf e}_1, & \text{for } i=1\\
   {\bf e}_{i-1}+2{\bf e}_i,& \text{for } i\in[d]\setminus\{1\}.
  
\end{cases}
\]

\begin{lemma}\label{lem:correct_hyp}
For {$d\in\mathbb{N}$}, each unit hypercube $H_d(r,1)$ centered at any point $r\in\IR^d$ contains at least one point of $\chi_d$.
\end{lemma}
\begin{proof}

Let \blue{$\Lambda_d'=\{p+{\bf v}\ |\ p\in\Lambda_d\text{ and }{\bf v}=\left(\frac{1}{2},\frac{1}{2},\ldots,\frac{1}{2}\right)\in\IR^d\}$} (for description of $\Lambda_2'$, see Fig.~\ref{fig:lambda'}). 
We need to show that $H_d(r,1)$ contains at least one point from $\chi_d$.
Note that there exists a point $s'\in\Lambda_d'$ such that $r$ belongs to the integer hypercube $H_d\left(s',\frac{1}{2}\right)$. Since for any $x\in H_d\left(s',\frac{1}{2}\right)$  the $dist_{\infty}(r,x)\leq 1$, the hypercube $H_d\left(s',\frac{1}{2}\right)$ is totally contained in $H_d\left(r,1\right)$.
\begin{clm}\label{clm:cont}
    For $d\in\mathbb{N}$, each integer hypercube $H_d\left(s,\frac{1}{2}\right)$ contains at least one integer point of $\chi_d$, where $s\in\Lambda_d'$.
\end{clm}
 \begin{proof}
 We  prove this claim using induction on $d$.
In the base case of the induction, for $d=1$ and 2, it is easy to observe from Fig.~\ref{fig:l1} that each integer hypercube $H_i\left(s,\frac{1}{2}\right)$ contains exactly one integer point of $\chi_d$.
For $d\in\{3,4,\ldots,i\}$, let us assume that the induction hypothesis holds. 
Now, to complete the proof, we need to show that the induction is also true for $d=i+1$. Now, for any $k\in \mathbb{Z}$, we define the hyperplane $P_k=\{x\in\IR^{i+1}\ |\ x_{i+1}=k\}$.
Notice that the integer hypercube $H_{i+1}\left(s,\frac{1}{2}\right)$ contains integer points from two consecutive hyperplanes $P_k$ and $P_{k+1}$ for some $k\in\mathbb{Z}$. As per the definition of $\chi_{i+1}$, exactly one of the hyperplanes $P_k$ and $P_{k+1}$ contains points from $\chi_{i+1}$, and the other one does not contain any point from $\chi_{i+1}$. Without loss of generality, let us assume that the hyperplane $P_k$ contains points of $\chi_{i+1}$. Notice that $\chi_{i+1}\cap P_k=\chi_i$ and $H_{i+1}\left(s,\frac{1}{2}\right)\cap P_k= H_{i}\left(s',\frac{1}{2}\right)$, for some $s'\in\Lambda_i'$. 
Due to the induction hypothesis, for any point $p\in\Lambda_{i}'$, the integer hypercube $H_{i}\left(p ,\frac{1}{2}\right)$ contains at least one point from $\chi_{i}$. Thus, $H_{i+1}\left(s,\frac{1}{2}\right)$ contains at least one point of $\chi_{i+1}$ from the hyperplane $P_k$.\end{proof}
 Due to Claim~\ref{clm:cont}, each integer hypercube $H_d\left(s,\frac{1}{2}\right)$  contains at least one integer point of $\chi_d$. As a result, the hypercube $H_d(r,1)$  contains at least one point of $\chi_d$. Hence, the lemma follows.
 \end{proof}
 
\begin{theorem}\label{cube_ub}
For hitting unit hypercubes in $\IR^3$ using points from $\mathbb{Z}^3$, there exists a deterministic online algorithm that  achieves a competitive ratio of at most~8.
\end{theorem}

\begin{proof}
For the sake of simplicity, throughout the proof, we use $\chi$ instead of $\chi_3$. For any $k\in  \mathbb{Z}$, we use $P_k$  to denote the plane parallel to $xy$-plane with $z$-coordinate value $k$. 
The projections of  planes  $P_{2k}, P_{2k+1}$ and $ P_{2k+2}$  over a rectangular region are depicted in Fig.~\ref{fig:l1}, Fig.~\ref{fig:l2} and  Fig.~\ref{fig:l3}, respectively. 
Note that $P_{2k+1}\cap\chi=\phi$. Also, observe that $P_{2k}\cap\chi$ and $P_{2k+2}\cap\chi$ are translated copy of each other by 1 unit in $y$-coordinate.
Algorithm $\BPA$ maintains a hitting set $\A$ consisting of points from $\chi$. On receiving a new input unit cube $\sigma$, if it is not hit by any of the points from $\A$ then the algorithm adds the best-point from $\chi$ lying inside $\sigma$ to the set $\A$. Correctness of the algorithm follows from Lemma~\ref{lem:correct_hyp}.

Let  $\I$ be the set of input unit cubes  presented to the algorithm.
Let $\OO$ be an offline optimal hitting set for $\I$.  
 Let $\A'=\A\setminus \{\A\cap \OO\}$ and  $\OO'=\OO\setminus \{\A\cap \OO\}$. 
Let $p \in \OO'$ be an integer point and let $\I_p\subseteq \I$ be the set of input unit cubes  containing the point $p$.
 Let $\A_{p}\subseteq \A'$ be the set of points our algorithm will place to hit explicitly when some unit cube in $\I_p$ arrives. In the following lemma, we prove that the cardinality of $\A_p$ is bounded by~8.
 Since $\A'=\cup_{p\in \OO'}\A_p$, we have $|\A'|\leq\sum_{p\in \OO'}|\A_p|\leq 8\times |\OO'|$. Note that $\frac{|\A'| }{|\OO'|}\leq 8$ implies $\frac{|\A|}{|\OO|}\leq 8$.
Thus, the competitive ratio of our algorithm is at most~8. \end{proof}
\begin{lemma}\label{cube}
$|\A_p|\leq 8$.
\end{lemma}
\begin{proof}
 Observe that the center of each $\sigma \in \I_p$ lies in the region $H_3(p,1)$, and to hit unit cubes of $\I_p$, our algorithm places integer points from $\chi(H_3(p,2))$.
 Therefore, $\A_p$ contains points from $\chi(H_3(p,2))$.
Let $p(z)$ be the z-coordinate value of the point $p$.
 Note that the cube $H_3(p,2)$ contains integer points only from five planes, 
 namely, $P_{p(z)+2},P_{p(z)+1},P_{p(z)},P_{p(z)-1}$ and $P_{p(z)-2}$. 
 As per the definition of $\chi$, if $p(z)$ is odd then $P_{p(z)+2}\cap\chi,\ P_{p(z)}\cap\chi,\ P_{p(z)-2}\cap\chi$ are empty, otherwise $P_{p(z)+1}\cap\chi$, $P_{p(z)-1}\cap\chi$ are empty.

\begin{figure}[htbp]
          \centering
          \hfill
     \begin{subfigure}[b]{0.24\textwidth}
         \centering
         \includegraphics[width=28 mm]{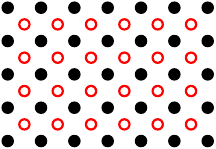}
         \caption{}
         \label{fig:lambda'}
     \end{subfigure}
     \hfill
          \begin{subfigure}[b]{0.24\textwidth}
         \centering
    \includegraphics[width=28 mm]{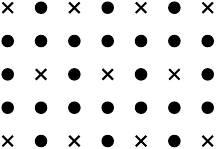}
    \caption{}
    \label{fig:l1}
     \end{subfigure}
     \hfill
     \begin{subfigure}[b]{0.24\textwidth}
         \centering
         \includegraphics[width=28 mm]{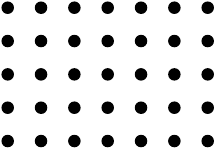}
         \caption{}
         \label{fig:l2}
     \end{subfigure}
       \begin{subfigure}[b]{0.24\textwidth}
         \centering
         \includegraphics[width=28 mm]{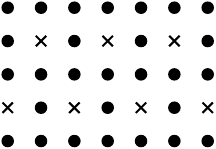}
         \caption{}
         \label{fig:l3}
     \end{subfigure}
     \caption{(a) The points of $\Lambda$ and $\Lambda'$ are represented in black and red color, respectively. {The projections of  planes  over a rectangular region (b) $P_{2k}$, (c) $P_{2k+1}$ and (d) $P_{2k+2}$.}}
   \label{fig:plane_part}
       \end{figure}
    
      \begin{figure}[htbp]
     \centering
     \hfill
     \begin{subfigure}[b]{0.33\textwidth}
         \centering
    \includegraphics[width=32 mm]{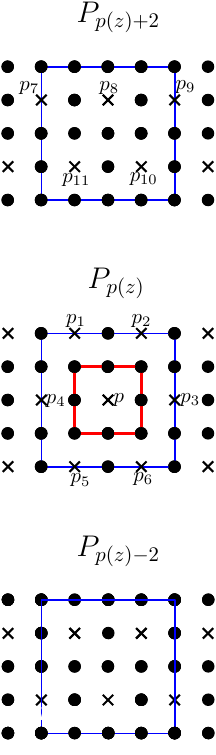}
    \caption{}
    \label{case-1_1}
     \end{subfigure}
     \hfill
     \begin{subfigure}[b]{0.33\textwidth}
         \centering
         \includegraphics[width=32 mm]{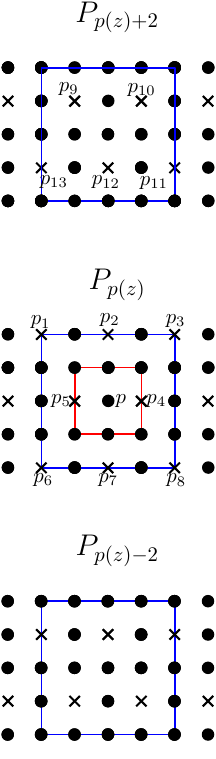}
         \caption{}
         \label{case-2_1}
     \end{subfigure}
     \hfill
       \begin{subfigure}[b]{0.32\textwidth}
         \centering
         \includegraphics[width=32 mm]{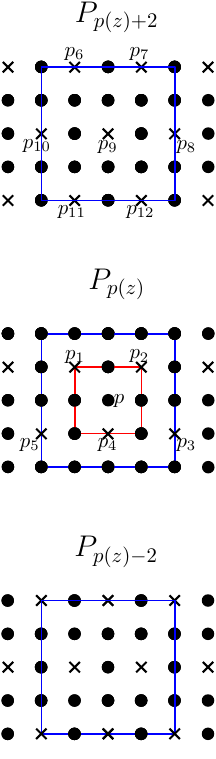}
         \caption{}
         \label{case-3_1}
     \end{subfigure}

         \caption{Illustration of Case~1: Here, boundaries of $H_3(p,1)$ and $H_3(p,2)$ are marked in red and blue colors, respectively, (a) Case~1.1, (b) Case~1.2 and (c) Case~1.3.
         }
   \label{fig:hq_ub_main}
\end{figure}

      \begin{figure}[htbp]
     \centering
    \hfill
     \begin{subfigure}[b]{0.49\textwidth}
         \centering
    \includegraphics[width=32 mm]{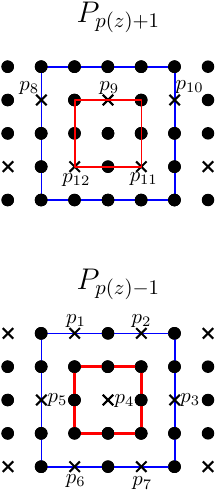}
    \caption{}
   \label{case-4_2}
     \end{subfigure}
     \hfill
     \begin{subfigure}[b]{0.5\textwidth}
         \centering
         \includegraphics[width=32 mm]{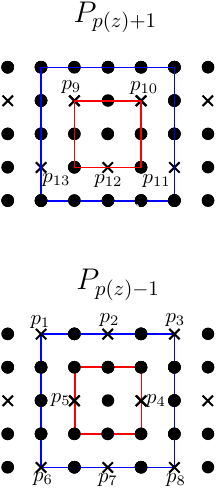}
         \caption{}
         \label{case-5_2}
     \end{subfigure}
         \caption{Illustration of Case~2: Here, boundaries of $H_3(p,1)$ and $H_3(p,2)$ are marked in red and blue color, respectively, (a) Case~2.1 and (b) Case~2.2.}
   \label{fig:hq_ub_main_b}
\end{figure}
\noindent
\textbf{Case 1 :} $p(z)$ is even. In this case, $|\chi(H_3(p,1))|$ is either 1,2 or 3. Depending on the value of $|\chi(H_3(p,1))|$, we have the following three subcases.\\ 
\textbf{Case 1.1 :} $|\chi(H_3(p,1))|=1$. Observe that, in this case, $p\in\chi$.
Representative figures of planes  $P_{p(z)-2}$, $P_{p(z)}$ and $P_{p(z)+2}$ intersecting the cube $H_3(p,2)$ are  shown in Fig.~\ref{case-1_1}.
Observe that  $H_3(p,2)$ contains 17 integer points of $\chi$ including $p$. As per the definition of $\A_p$, we know $p\notin \A_p$. 
Remember that any unit cube $\sigma\in\I_p$  contains  the point $p$. For  any point $p'\in P_{p(z)-2}$, we have $p'\prec p$.  Thus, our algorithm does not add $p'$ to $\A_p$. As a result, none of the points of $P_{p(z)-2}$ are in $\A_p$.
Similarly, it is easy to see $p_4,p_5,p_6\prec p$.  Thus, points $p_4,p_5, p_6 \notin\A_p$. 
Now, consider  any unit cube $\sigma_1\in\I_p$ that contains the point $p_7$. Note that $\sigma_1$ must also contain $p_8$, and $p_7\prec p_8$. Therefore, our algorithm does not add $p_7$ to $\A_p$ upon the arrival of $\sigma_1$.  Putting all these together, all five points from $P_{p(z)-2}$, $\{p,p_4,p_5,p_6\}$ from $P_{p(z)}$ and $p_{7}$ from $P_{p(z)+2}$ are not in $\A_p$.  Therefore, we have $|\A_p|\leq 7$.
\\
\textbf{Case 1.2 :} $|\chi(H_3(p,1))|=2$. In this case, it is easy to observe that the plane $P_{p(z)}$  contains both the  points of $\chi(H_3(p,1))$.
Representative figures of planes $P_{p(z)-2}$, $P_{p(z)}$ and $P_{p(z)+2}$ intersecting the cube $H_3(p,2)$ are  depicted in Fig.~\ref{case-2_1}. It is easy to see that the cube  $H_3(p,2)$ contains 18 integer points of $\chi$.
 Notice that any unit cube $\sigma\in\I_p$ that contains $p$ must also contain either $p_4,p_5$ or both from the plane $P_{p(z)}$. Observe that for  any point $p'\in P_{p(z)-2}$, we have $p'\prec p_4, p_5$.  As a result, none of the points from the plane $P_{p(z)-2}$ are in $\A_p$.
Similarly, since $p_6,p_7\prec p_4,p_5$, we know that $p_6, p_7 \notin\A_p$.
 Observe that any unit cube  that contains $p$ and $p_1$ (respectively, $p_8$) must also contain the point $p_2$ (respectively $p_4$). Since $p_1\prec p_2$ and $p_8\prec p_4$, we know that $p_1, p_8 \notin\A_p$. 
 Similarly, any unit cube that contains the point $p$ and $p_{13}$, must also contain $p_{12}$, and we have $p_{13}\prec p_{12}$. Thus, $p_{13}\notin\A_p$. 
Combining all of these, we know that  all five points from $P_{p(z)-2}$, $\{p_1,p_6,p_7,p_8\}$ from $P_{p(z)}$ and $p_{13}$ from $P_{p(z)+2}$ are not in $\A_p$. Hence, we have $|\A_p|\leq 8$. \\
\noindent
\textbf{Case 1.3:} $|\chi(H_3(p,1))|=3$. 
Planes $P_{p(z)-2}$, $P_{p(z)}$ and $P_{p(z)+2}$  intersecting the cube $H_3(p,2)$ are depicted in Fig.~\ref{case-3_1}. 
Observe that the cube $H_3(p,2)$ contains only 19 integer points of $\chi$.
Notice that any unit cube that contains $p$ must contain either $p_1,p_2$ or $p_4$ from $p_{p(z)}$. On the other hand, for any point $p'\in P_{p(z)-2}$,  we know that $p'\prec p_1,p_2,p_4$. Thus, our algorithm does not add any of the points from the plane $P_{p(z)-2}$ to $\A_p$. Any unit cube that contains $p$ and $p_5$ must also contain $p_4$, and we know that $p_5\prec p_4$.  Therefore, $p_5 \notin \A_p$. Now, observe that any unit cube that contains $p$ and  some point from $P_{p(z)+2}$ must also contain the point $p_9$. Since $p_{10},p_{11}, p_{12} \prec p_9$, we know that $p_{10},p_{11}, p_{12} \notin \A_p$. 
After putting all of these together, we know that all seven points from $P_{p(z)-2}$, $p_5$ from $P_{p(z)}$ and $\{p_{10},p_{11},p_{12}\}$ from $P_{p(z)+2}$ are not in $\A_p$. Thus, we have $|\A_p|\leq 8$. \\

\noindent
\textbf{Case 2 :} $p(z)$ is odd. In this case  $|\chi(H_3(p,1))|$ is either 4 or 5. Depending on the value of $|\chi(H_3(p,1))|$, we have the following two subcases.\\
\textbf{Case 2.1:} $|\chi(H_3(p,1))|=4$. 
Representative figures of planes  $P_{p(z)-1}$ and $P_{p(z)+1}$ intersecting the cube $H_3(p,2)$ are  depicted in Fig.~\ref{case-4_2}.
Observe that the cube $H_3(p,2)$ contains only 12 integer points of $\chi$. 
Observe that any unit cube that contains $p$ and 
any of the points in $\{p_5,p_6,p_7\}$ must also contain  $p_4$, and we know that  $p_5,p_6,p_7\prec p_4$. Similarly, any unit cube that contains $p$ and $p_8$ must also contain $p_9$, and we have $p_8\prec p_9$. As a result,  $p_5,p_6,p_7$ and $p_{8}$  are not in $\A_p$. Hence, $|\A_p|\leq 8$. \\
\textbf{Case 2.2:} $|\chi(H_3(p,1))|=5$. 
Planes $P_{p(z)-1}$ and $P_{p(z)+1}$ intersecting the cube $H_3(p,2)$ are shown in Fig.~\ref{case-5_2}. 
The cube $H_3(p,2)$ contains 13 integer points of $\chi$.
Similar to Case 1.2, one can observe that none of the points $p_1,p_6,p_7,p_8$ from $P_{p(z)-1}$ and $p_{13}$ from $P_{p(z)+1}$ are in $\A_p$. As a result, $|\A_p|\leq 8$. 
\end{proof}

\begin{theorem}\label{square_ub}
For hitting unit squares using points from $\mathbb{Z}^2$, there exists a deterministic online algorithm that  achieves a competitive ratio of at most~4.
\end{theorem}
\begin{proof}

The proof is similar to Theorem~\ref{cube_ub}. Here,
the integer lattice $\Lambda=\{q {\bf e}_1+r {\bf e}_2| q,r\in \mathbb{Z}\}$ in $\IR^2$ is  generated by standard unit vectors ${\bf e}_1$ and ${\bf e}_2$.
Consider a subset $\chi\subset\Lambda$ defined as follows: $\chi=\{q {\bf u}+r {\bf v} |  q,r\in \mathbb{Z}\}$, where ${\bf u}=2{\bf e}_1$ and ${\bf v}={\bf e}_1+2{\bf e}_2$.  Next, we prove that $|\A_p|\leq 4$ in this case. Hence, the competitive ratio of algorithm $\BPA$ is at most 4.

Note that a unit square centered at an integer point contains at least one and at most 3 points of $\chi$ (see plane $P_{p(z)}$ in Fig.~\ref{case-1_1},~\ref{case-2_1} and~\ref{case-3_1}). Hence, we have the following three cases. \\
\textbf{Case 1:} $|\chi(H_2(p,1))|=1$. Note that $H_2(p,2)$ contains 7 integer points of $\chi$ including $p$ (see plane $P_{p(z)}$ in Fig.~\ref{case-1_1}). With the similar argument of Case 1.1 of Lemma~\ref{cube},   one can easily notice that none of the points $\{p_4,p_5,p_6\}$ are in $\A_p$.
As a result $|\A_p|\leq 4$.\\
\textbf{Case 2:} $|\chi(H_2(p,1))|=2$. Observe that $H_2(p,2)$ contains only eight integer points of $\chi$ (see plane $P_{p(z)}$ in Fig.~\ref{case-2_1}). With the similar argument of Case 1.2 of Lemma~\ref{cube}, it is easy to observe that none of the points $\{p_1,p_6,p_7,p_8\}$ are in $\A_p$. As a result, $|\A_p|\leq 4$.\\
\textbf{Case 3:} $|\chi(H_2(p,1))|=3$. Note that $H_2(p,2)$ contains only five integer points of $\chi$ (see plane $P_{p(z)}$ in Fig.~\ref{case-3_1}). With the similar argument of Case 1.3 of Lemma~\ref{cube}, one can see that $p_5$ is not in $\A_p$.
Hence, $|\A_p|\leq4$. \end{proof}

\subsection{Unit Hypercubes in $\mathbb{R}^d$ ($d\geq 3$)}\label{sec:hyp}

We first present some concepts that we will utilize to analyze the algorithm we propose for unit hypercubes in $\mathbb{R}^d$, where $d\geq 3$. Let $\cal F$ be the family of all possible unit hypercubes in $\IR^d$.
Any pair of unit hypercubes $\sigma_i$ and $\sigma_j$ in  $\cal F$ are said to be \emph{related} if $\Q(\sigma_i)=\Q(\sigma_j)$, in other words, each of them contains the same set of integer points. 
So, we have an equivalence relation on $\cal F$ where each \emph{equivalence class} corresponds to a set $S$ of unit hypercubes such that each $\sigma\in S$ contains the same set of integer points.

Let $\sigma$ be a unit hypercube centered at a point $c\in\IR^d$. Partition $[d]$ into two sets $\K_1$ and $\K_2$ such that 
 for each $i\in\K_1$, the value of $c(x_i)$ is non-integer and for each $i\in\K_2$, the value of $c(x_i)$ is integer. Let $r\in{\cal Q}(\sigma)$ be  an integer point. For any $i\in\K_1$, the value of $r(x_i)$ can be one from the two possible values: $\{\lfloor c(x_i) \rfloor,\lceil c(x_i) \rceil\}$, and for any $i\in\K_2$, the value of $r(x_i)$ can be one from the three possible values $\{c(x_i)-1,c(x_i),c(x_i)+1\}$.  Hence, ${\cal Q}(\sigma)$ contains exactly $2^{|\K_1|}3^{|\K_2|}$ integer points. 
The following lemma is an important ingredient for classifying the equivalence classes.

\begin{lemma}\label{Lemma_1}
Let $\sigma_1$ and $\sigma_2$ be two  unit hypercubes centered at $c_1$ and $c_2$ in $\IR^d$, respectively. Both $\sigma_1$ and $\sigma_2$ contain the same set of integer points if and only if 
$[d]$ can be partitioned into two sets $\K_1$ and $\K_2$
 such that 
\begin{itemize}
\item  
for each $i\in \K_1$, the value of the $i$th coordinate of $c_1$ and $c_2$ is non-integer 
and $\lfloor c_1(x_i)\rfloor=\lfloor c_2(x_i)\rfloor$,

\item for each $i\in\K_2$, the value of 
 the $i$th coordinate  of  $c_1$ and $c_2$ is same,  i.e., $c_1(x_i)=c_2(x_i)$.
\end{itemize}
\end{lemma}
\begin{proof}
For the forward direction, we prove the contrapositive statement:  ``If there exists some $i\in[d]$ such that $\lfloor c_1(x_i)\rfloor\neq\lfloor c_2(x_i)\rfloor$, then ${\cal Q}(\sigma_1)\neq{\cal Q}(\sigma_2)$". 
Since  $\lfloor c_1(x_i) \rfloor\neq\lfloor c_2(x_i)\rfloor$,   without loss of generality, let us assume that $c_1(x_i)>c_2(x_i)$.
 Let us consider a point $r$ whose $j$th coordinate 
 is defined as $r(x_j)=\lfloor c_1(x_j) \rfloor +1$, for all $j\in [d]$.
  It is easy to note that the point $r\in {\cal Q}(\sigma_1)$.
Since the difference between $c_2(x_i)$ and $r(x_i)$ is more than one, the point $r\notin {\cal Q}(\sigma_2)$. Hence, we have ${\cal Q}(\sigma_1)\neq{\cal Q}(\sigma_2)$.

Now, we consider the converse part. Assume that for each $i\in\K_1$, $\lfloor c_1(x_i)\rfloor=\lfloor c_2(x_i)\rfloor$, and for each $i\in\K_2$, $c_1(x_i)=c_2(x_i)$. We need to prove that ${\cal Q}(\sigma_1)={\cal Q}(\sigma_2)$. First, we prove that $ {\cal Q}(\sigma_1)\subseteq {\cal Q}(\sigma_2)$. The other case ($ {\cal Q}(\sigma_2)\subseteq {\cal Q}(\sigma_1)$) is symmetric in nature. Let $r_1\in{\cal Q}(\sigma_1)$. 
 For each $i\in\K_1$, $r_1(x_i)$ has only two possibilities from $\{\lfloor c_1(x_i) \rfloor,\lceil c_1(x_i)\rceil\}$. Since  $\lfloor c_1(x_i)\rfloor=\lfloor c_2(x_i)\rfloor$, the difference between $c_2(x_i)$ and $r_1(x_i)$ is at most one. For each $i\in\K_2$, $r_1(x_i)$ has three possibilities from $\{c_1(x_i)-1,c_1(x_i),c_1(x_i)+1\}$. Since $c_1(x_i)=c_2(x_i)$,  the difference between $c_2(x_i)$ and $r_1(x_i)$ is at most one. As a result, $dist_{\infty}(c_2,r_1)\leq 1$. Hence, we have $r_1\in{\cal Q}(\sigma_2)$. \end{proof}

Using the above lemma, we prove the next two lemmas that will play an important role in analysing our algorithm.
We have $d+1$ types of equivalence classes depending on the number of integer points they cover.
We refer to an equivalence class that contains exactly $2^k3^{d-k}$ integer points as \emph{an equivalence class of  Type-($k$)}, where $k\in[d]\cup\{0\}$.
By careful observation, one can note the following.
\begin{lemma}\label{claim:subset}
Let $\sigma$ be a unit hypercube in $\mathbb{R}^d$, centered at a point $c\in\IR^d$, belonging to some equivalence class of  Type-($k$), where $k\in[d-1]\cup\{0\}$. There exists a set ${\mathbb S}_{\sigma}$ of distinct $2^{(d-k)}$ equivalence classes of  Type-($d$) such that ${\cal Q}(\sigma)=\cup_{\sigma'\in{\mathbb S}_{\sigma}}{\cal Q}\left(\sigma'\right)$.
\end{lemma}
\begin{proof}
Let  $\K_1$ and $\K_2$ be the partition of $[d]$ depending upon the noninteger coordinates of the center $c$ such that $|\K_1|=k$.  For the sake of simplicity, let us assume that  the first $d-k$ indices of  $[d]$  belong to $\K_2$, and  the remaining $k$ indices belong to $\K_1$.
Now, we construct $2^{d-k}$ many  hypercubes.
Let $0\leq t < 2^{d-k}$ be an integer, and  
let $t_1t_2\ldots t_{d-k}$ be the  binary representation of $t$.
Let $\sigma_t$ be a unit hypercube, centered at $c_t$  such that
for each $i\in\K_1$, $c_t(x_i)$ is equal to $c(x_i)$, and 
for each $i\in\K_2$, the $i$th coordinate of $c_t$ is defined as follows.
\begin{equation*}\label{c_p}
c_t(x_i)=
    \begin{cases}
   c(x_i)+\epsilon,& \text{if $t_i=0$}\\
     c(x_i)-\epsilon, & \text{otherwise (i.e., if $t_i=1$)},
    \end{cases}      
\end{equation*}
 where $c$ is the center of the hypercube $\sigma$ belonging to an equivalence class of  Type-($k$) and  $0<\epsilon<1$ is a fixed arbitrary constant close to zero. 
Let ${\mathbb S}_{\sigma}=\{\sigma_0,\sigma_1,\ldots,\sigma_{2^{d-k}-1}\}$.
 Due to Lemma~\ref{Lemma_1}, it is easy to observe that each hypercube in ${\mathbb S}_{\sigma}$ belongs to  a distinct equivalence class. Therefore, ${\mathbb S}_{\sigma}$ consists of $2^{d-k}$ distinct equivalence classes of  Type-($d$). 
 
Now, we show that  ${\cal Q}(\sigma')\subseteq{\cal Q}(\sigma)$, for each $\sigma'\in{\mathbb S}_{\sigma}$ centered at $c'$.
Let $r'\in {\cal Q}(\sigma')$ be an integer point. For each $i\in\K_1$, $r'(x_i)$ has only two possibilities from $\{\lfloor c(x_i) \rfloor,\lceil c(x_i)\rceil\}$. Therefore, the difference between $c(x_i)$ and $r'(x_i)$ is at most one. On the other hand, for each $i\in\K_2$, $r'(x_i) $ has two possibilities from $\{\lfloor c'(x_i)\rfloor,\lceil c'(x_i)\rceil\}$. 
Since the value of $c'(x_i)$ is either $c(x_i)+\epsilon$ or $c(x_i)-\epsilon$, it is easy to observe that 
the difference between $c(x_i)$ and $r'(x_i)$ is at most one. As a result, $r'$ belongs to ${\cal Q}(\sigma)$. Hence, for any $\sigma'\in{\mathbb S}_{\sigma}$, ${\cal Q}(\sigma')\subseteq{\cal Q}(\sigma)$. Therefore, $\cup_{\sigma'\in{\mathbb S}_{\sigma}} {\cal Q}(\sigma')\subseteq {\cal Q}(\sigma)$.

Now, we prove that $ {\cal Q}(\sigma)\subseteq\cup_{\sigma'\in{\mathbb S}_{\sigma}}{\cal Q}(\sigma')$. Let $r\in{\cal Q}(\sigma)$  be an integer point. Now, we need to construct some $c'$ such that $\sigma'\in{\mathbb S}_{\sigma}$, centered at $c'$, contains the point $r$.
 Let us define, for each $i\in\K_1$, $c'(x_i)$ to be equal to $c(x_i)$. Note that, for each $i\in\K_2$, $r(x_i)$ has three possibilities from $\{c(x_j)-1,c(x_j),c(x_j)+1\}$. For each $i\in\K_2$, the $i$th coordinate of $c'$ is constructed as follows.
\begin{equation*}\label{c_q}
c'(x_i)=
    \begin{cases}
    c(x_i)-\epsilon,& \text{if $r(x_i)=c(x_i)-1$}\\
     c(x_i)+\epsilon, & \text{otherwise}.
    \end{cases}      
\end{equation*}
Observe that the hypercube $\sigma'$ defined above belongs to ${\mathbb S}_{\sigma}$.
Since the distance (under $L_{\infty}$ norm) between $r$ and $c'$ is at most one, the hypercube $\sigma'$ contains the point $r$.
Hence, $r\in\cup_{\sigma'\in\s}{\cal Q}(\sigma')$. As a result, we have ${\cal Q}(\sigma)\subseteq\cup_{\sigma'\in{\mathbb S}_{\sigma}}{\cal Q}(\sigma')$.
 Therefore, ${\cal Q}(\sigma)=\cup_{\sigma'\in{\mathbb S}_{\sigma}}{\cal Q}(\sigma')$. \end{proof}

 \begin{lemma}\label{hit_hyp}
Each integer point $p\in\mathbb{Z}^d$ is contained in exactly $2^d$ distinct equivalence classes of  Type-($d$), where $d\in\mathbb{N}$.
\end{lemma}

\begin{proof}
Let $\sigma$ be a hypercube, centered at $c$, belonging to an equivalence class of  Type-($d$).
Due to Lemma~\ref{Lemma_1}, any hypercube $\sigma'$ belongs to the same equivalence class of $\sigma$ if and only if for each $i\in[d]$, $c'(x_i)$ lies in the open interval $\left(\lfloor c(x_i)\rfloor,\lceil c(x_i)\rceil\right)$. In other words, the center of each of these hypercubes $\sigma'$ lies in the interior of an integer hypercube that contains the point $c$.
Therefore, the interior of each integer hypercube represents centers of hypercubes belonging to an equivalence class of  Type-($d$).
\noindent
Let $H_p$ be a $d$-dimensional unit hypercube centered at an integer point $p$. Note that all the centers of unit hypercubes containing the point $p$ must lie in $H_p$. Observe that the  unit hypercube $H_p$ contains  exactly $2^d$ many integer hypercubes. This implies that the integer point $p$ is contained in exactly $2^d$ many equivalence classes of  Type-($d$). \end{proof}

To obtain the upper bound,  for $d\geq 3$,   we propose an $O(d^2)$-competitive {algorithm, $\textsc{Randomized-}$ $\textsc{Iterative-Reweighting}$}, that is similar in nature to an algorithm from~\cite{DumitrescuT22}, which was presented for covering integer points using integer hypercubes in the online setup.

\textit{\textbf{Algorithm $\RIR$:}} Let $\I$ be the set of hypercubes presented to the algorithm and $\A$ be the set of points chosen by our algorithm such that each hypercube in $\I$ contains at least one point from $\A$. The algorithm maintains two disjoint sets $\A_1$ and $\A_2$ such that $\A= \A_1 \cup \A_2$. The algorithm also maintains another set $\B$ of points for bookkeeping purposes; initially, each of the set $\I,\A\text{ and }\B$ are empty. A weight function  $w$ over all integer points is also maintained by the algorithm; initially, $w(p)=3^{-(d+1)}$, for all points $p \in \mathbb{Z}^d$. One iteration of the algorithm is described below.

Let $\sigma$ be a new hypercube; update $\I = \I \cup\{\sigma\}$. Note that $|\Q(\sigma)|$ is at least $2^d$ and at most $3^d$.
\begin{itemize}
\item[1.] If the hypercube $\sigma$ contains any point from $\A$, then do nothing.
\item[2.] Else if the hypercube $\sigma$ contains any point from $\B$, then let $p\in \B\cap\QH$ be an arbitrary point, and update $\A_1 = \A_1 \cup\{p\}$.
\item[3.] Else if $\sum_{p\in \QH} w(p)\geq 1$, then let $p$ be an arbitrary point in $\QH$, and update $\A_2 = \A_2 \cup\{p\}$.
\item[4.] Else, the weights give a probability distribution on $\QH$. Successively choose points from $\QH$ at random with this distribution in $\lceil\frac{5d}{2}\rceil$ independent trails and add them to $\B$. Let $p\in \B\cap\QH$ be an arbitrary point, and update $\A_1 = \A_1 \cup\{p\}$. Triple the weight of every point in $\QH$.
\end{itemize}

Now, we analyze the performance  of the above algorithm.
\begin{theorem}\label{hyp_ub}
 {The algorithm $\RIR$} achieves a competitive ratio of at most~$O(d^2)$ for hitting unit hypercubes using points in $\mathbb{Z}^d$, where $d\geq3$.
\end{theorem}
\begin{proof}
Let  $\I$ be the set of $n$ hypercubes presented to our algorithm. Let $\OO$ be an offline optimum hitting set for $\I$.  Note that our algorithm creates two disjoint sets $\A_1\text{ and }\A_2$ such that $\A=\A_1\cup\A_2$ is a hitting set for $\I$.
From the description of the algorithm, it is easy to follow that $\A_1\subseteq \B$. We prove that $\mathbb{E}[|\B|]=O(d^2|\OO|)$ and $\mathbb{E}[|\A_2|]=O(|\OO|)$. This immediately implies that $\mathbb{E}[|\A|]\leq\mathbb{E}[|\A_1|] +\mathbb{E}[|\A_2|]\leq \mathbb{E}[|\B|]+\mathbb{E}[|\A_2|]=O(d^2|\OO|)$.

First, consider $\mathbb{E}[|\B|]$. Note that in the set $\B$, new points are added only in step 4 of the algorithm. In this case, the algorithm adds at most $\lceil\frac{5d}{2}\rceil$ points (independently) in $\B$ and triples the weight of every point in $\QH$. Each hypercube $\sigma\in \I$ contains some point $p\in \OO$. Initially, the weight of $p$ is $3^{-(d+1)}$, and it will never exceed 3. Since $p\in \QH$ and its weight before the last tripling must have been at most 1 in step 4 of the algorithm, its weight is tripled in at most $d+2$ iterations. Consequently, the algorithm invokes step 4  of the algorithm in at most $(d+2)|\OO|$ iterations. In each such iteration, the algorithm adds at most $\lceil\frac{5d}{2}\rceil$ points (independently) in the set $\B$. Therefore, we have $|\B|\leq \lceil\frac{5d}{2}\rceil(d+2)|\OO|=O(d^2|\OO|)$.

Next, we consider $\mathbb{E}[|\A_2|]$. Note that in the set $\A_2$, new points are added only in step 3 of the algorithm. In this case,  when a hypercube $\sigma$ arrives,  none of the points of $\QH$ is in $\B$ and $\sum_{p\in \QH} w(p)\geq 1$, and the algorithm increments the cardinality of the set $\A_2$ by one.
At the beginning of the algorithm, we have $W_{initially}=\sum_{p\in \QH} w(p)=\sum_{p\in\QH} 3^{-(d+1)}\leq 3^d  3^{-(d+1)}= \frac{1}{3}$. Suppose that the weights of the points in $\QH$ are increased in $k$ iterations (starting from the beginning of the algorithm), and the sum of weights of points in $\QH$ is increased by $\delta_1,\delta_2,\ldots,\delta_k>0$. When $\sigma$ arrives, the sum of the weights of all the points in $\Q(\sigma)$ is $W_{now}=W_{initially}+\sum_{i=1}^k\delta_i\geq 1$  and we know $\ W_{initially}\leq \frac{1}{3}$. This
implies that $\sum_{i=1}^k\delta_i\geq \frac{2}{3}$. For every $i\in[k]$, the sum of weights of some points in $\QH$, say $Q_i\subset \QH$ is increased by $\delta_i$ in step 4 of the algorithm. Since the weights are tripled, the sum of the weights of these points was $\frac{\delta_i}{2}$ at the beginning of that iteration. The algorithm added a point from $Q_i$ to $\B$ with probability at least $\frac{\delta_i}{2}$ in one random draw, which was repeated $\lceil\frac{5d}{2}\rceil$ times independently. As a result, the probability that the algorithm does not add any point from $Q_i$ to the set $\B$ is at most $\left(1-\frac{\delta_i}{2}\right)^{\lceil\frac{5d}{2}\rceil}$. The probability that none of the points of $\QH$ are added to $\B$ before the arrival of $\sigma$ is at most $ \prod_{i=1}^k \left(1-\frac{\delta_i}{2}\right)^{\lceil\frac{5d}{2}\rceil}\leq e^{-\lceil\frac{5d}{2}\rceil\sum_{i=1}^k\frac{\delta_i}{2}}\leq e^{-\frac{5d}{4}\sum_{i=1}^k\delta_i}\leq e^{-\frac{5d}{6}}$.
Since $\I$ is the set of hypercubes presented to the algorithm, step 3 of the algorithm can be invoked at most $|\I|$ times $e^{-\frac{5d}{6}}$. As a result, we have $\mathbb{E}[|\A_2|]\leq |\I|  e^{-\frac{5d}{6}}$.
Note that this is a very loose upper bound. 
Let $N$ be the set of distinct  equivalence classes containing all the hypercubes in $\cal I$.
Observe that if the algorithm hits  one hypercube from an equivalence class, then the algorithm executes only step 1 for all subsequent hypercubes coming from the same equivalence class. 
Therefore, step 3 of the algorithm can be invoked at most $|N|e^{-\frac{5d}{6}}$ times. We can further improve this bound  as follows.

Let $\sigma\in \I$. According to Lemma~\ref{claim:subset}, 
we have a set ${\mathbb S}_{\sigma}$ of equivalence classes of type-($d$) such that ${\cal Q}(\sigma)=\cup_{\sigma'\in{{\mathbb S}_{\sigma}}}{\cal Q}\left(\sigma'\right)$. 
Observe that  if  some hypercube $\sigma$ arrives and our algorithm needs to place a hitting point $p$ for it, then it implies that none of the hypercubes belonging to ${\mathbb S}_{\sigma}$ arrived before $\sigma$ to the algorithm.
Let $p\in {\cal Q}(\sigma')$ for some $\sigma' \in {\mathbb S}_{\sigma}$. Note that the point $p$ acts as a  hitting point for any  hypercube in $\I$ belonging to the same class of $\sigma'$. Not only that but $p$ also acts as a hitting point for all hypercubes $\sigma''\in \I$ such that $\sigma'\in {\mathbb S}_{\sigma''}$. Therefore, step 3 of the algorithm is invoked at most $|N_d|e^{-\frac{5d}{6}}$ times, where $N_d=\cup_{\sigma\in\I} {\mathbb S}_{\sigma}$.
 Hence, $\mathbb{E}[|\A_2|]\leq |N_d|e^{-\frac{5d}{6}}$.
Now, we  give an upper bound of $|N_d|$ in terms of $|\OO|$. Due to Lemma~\ref{hit_hyp}, we know that any arbitrary integer point $p\in\OO$ can be contained in at most  $2^d$ equivalence classes of Type-($d$) hypercubes. Thus, we have $|N_d|\leq 2^d  |\OO|)$.
 Since $\mathbb{E}[|\A_2|]\leq |N_d|  e^{-\frac{5d}{6}}$ and $|N_d|\leq 2^d |\OO|$, we have $\mathbb{E}[|\A_2|]\leq O\left(\left(\frac{2}{e^\frac{5}{6}}\right)^d |\OO| \right)\leq |\OO|$. Hence, the theorem follows. \end{proof}

\subsection{Lower Bound}
In this subsection, we present the lower bound of the competitive ratio for hitting unit hypercubes in $\IR^d$.

\begin{theorem}\label{hyp_lb}
The  competitive ratio of every deterministic online algorithm  for hitting hypercubes in $\mathbb{R}^d$ using points in $\mathbb{Z}^d$ is at least~$d+1$, where $d\in\mathbb{N}$.
\end{theorem}
\begin{proof}
 Let us consider a game between  two players: Alice and Bob. Here, Alice plays the role of an adversary, and Bob plays the role of an online algorithm. In each round of the game, Alice presents a new unit hypercube $\sigma\subset\mathbb{R}^d$ such that Bob needs to hit it by a new hitting point $h\in \mathbb{Z}^d$.
To prove the lower bound of the competitive ratio, we show by induction that Alice can  present a sequence of unit hypercubes  $\sigma_1,\sigma_2,\ldots,\sigma_{d+1}\subset \mathbb{R}^d$ adaptively, depending on the position of hitting points placed by Bob such that Bob needs to place $d+1$ integer points $\{h_1,h_2,\ldots,h_{d+1}\}$; whereas an offline optimum needs just one integer point.
Let $\sigma_{1}$ be a hypercube presented by Alice in the first round of the game.
For the sake of simplicity, we assume that the center of $\sigma_1$ is the origin.
 For $i=1,\ldots, {d+1}$, we maintain the following two invariants:
\begin{itemize}
    \item The hypercube $\sigma_i\subset \mathbb{R}^d$ does not contain any of the previously placed hitting point $h_j\in \mathbb{Z}^d$, for 
$j\in[i-1]$.
    \item The common intersection region $Q_i=\cap_{j=1}^{i}\sigma_j$  contains  $3^{(d-i+1)}$ integer points. 
\end{itemize}

  For $i=1$, the first invariant trivially holds. Since the unit hypercube $\sigma_1$ is centered at the origin, each coordinate of any integer point $p\in \sigma_1$ has three possible values from $\{-1,0,1\}$. As a result, the unit hypercube $\sigma_1$ contains $3^d$ integer points. Thus, the second invariant also holds.
  
  At the beginning of the round $i$ (for $i=2,\ldots,d$), let us assume that both invariants hold. Now, we only need to show that the induction is true for $i=d+1$. Let us define a translation vector ${\bf v}_i\in \mathbb{R}^{d}$ as follows:
${\bf v}_i=(s(1)(1+\epsilon),s(2)(1+\epsilon),\ldots,s(i-1)(1+\epsilon),0,\ldots,0)$, where $0<\epsilon<\frac{1}{2}$ is an arbitrary constant close to zero, and for any $j\in[i-1]$, we have
  \[
  s(j)= 
\begin{cases}
    +1,& \text{if}\ h_{j}(x_j)\leq 0,\text{ where}\ h_{j}(x_j)\text{ is $jth$ coordinate of $h_{j}$}, \\
    -1,              & \text{otherwise.}
\end{cases}
\]
We define the hypercube $\sigma_i=\sigma_1 +{\bf v}_i$. For any  $j\in[i-1]$, due to the definition of the $j$th component of the translation vector ${\bf v}_i$, the hypercube $\sigma_i$ does not contain the point $h_j$. Hence, the first invariant is maintained. 
 Let us count the number of integer points contained in $\sigma_1\cap\sigma_i$. Consider any integer point $p\in \sigma_1\cap\sigma_i$. 
Since $\sigma_i=\sigma_1 +{\bf v}_i$ and $\sigma_1$ is centered at the origin, $\sigma_i$ is centered at ${\bf v}_i$.
As a result, for any $j\in[i-1]$,  the $j$th coordinate of  $p$ is fixed at $s(j)$. The value of each of the remaining $(d-i+1)$ coordinates of $p$ has three possibilities from $\{-1,0,1\}$. Therefore, $\sigma_1\cap\sigma_i$  contains $3^{(d-i+1)}$ integer points. Because of the above argument, observe that all the integer points that belong to $\sigma_1\cap\sigma_i$ are also contained in  $\sigma_1\cap\sigma_j$, where $j\in[i-1]$. Hence, $Q_i$ contains $3^{(d-i+1)}$ integer points. \end{proof}


\section{Unit Covering Problem}\label{sec:cover}
Recall that by interchanging the role of unit objects and points, one can formulate an equivalent  \emph{online unit covering problem}.
Here, the points belonging to $\mathbb{R}^d$ arrive one by one. Upon the arrival of an uncovered point, we need to cover it using a unit object having center in $\mathbb{Z}^d$. 
Similar to the online hitting set problem, here, the decision to add a unit object is  irrevocable, i.e., the online algorithm can not remove  any unit object in the future from the existing cover. 
The aim of the online unit covering problem is to minimize the number of unit objects to cover all the presented points.
Since the above-mentioned unit covering problem is an equivalent version of the online hitting set problem studied in this paper, all the results obtained for the online hitting set problem are also valid for the equivalent online unit covering problem.
We summarize the results obtained for the online unit covering problem as follows. First, we present the lower bound of the online covering problem.
Due to Theorems~\ref{hyp_lb} and~\ref{ball_lb}, we obtain the following.
\begin{corollary} The  competitive ratio of every deterministic online algorithm  for covering points in $\mathbb{R}^d$ using
    \begin{itemize}
    \item [(i)] unit hypercubes centered at $\mathbb{Z}^d$ is at least~$d+1$, where $d\in\mathbb{N}$.
    \item [(ii)] unit balls centered at $\mathbb{Z}^d$ is at least~$d+1$, where $d<4$.
    \end{itemize}
\end{corollary}
\noindent
Now, we present the obtained upper bounds.
Due to Theorems~\ref{thm:int}-~\ref{2d-balls}, we have the following.
\begin{corollary} 
    \begin{itemize}
    \item [(i)] For covering points in $\IR^d$ using unit hypercubes centered at $\mathbb{Z}^d$, where $d=1,2$ and 3, respectively, there exist  deterministic online algorithms having  competitive ratios at most 2,4 and 8, respectively.
   
     \item [(ii)] For covering points in $\IR^d$ using unit balls centered at $\mathbb{Z}^d$, where $d=2$ and 3, respectively, there exist deterministic online algorithms having competitive ratios of at most~4 and 14, respectively.
    \end{itemize}
\end{corollary}
Due to Theorems~\ref{hyp_ub} and~\ref{ball_ub}, we have the following.
\begin{corollary} For covering points in $\IR^d$ using
    \begin{itemize}
    \item [(i)] unit hypercubes centered at $\mathbb{Z}^d$, there exists a randomized algorithm whose competitive ratio is~$O(d^2)$, where $d\geq3$.
    \item [(ii)] unit balls centered at $\mathbb{Z}^d$, there exists a deterministic online algorithm whose competitive ratio is~$O(d^4)$, where $d\in\mathbb{N}$.
   
    \end{itemize}
\end{corollary}


\section{Conclusion}\label{Conclusion}
In this paper, we have considered the online hitting set problem for unit balls and unit hypercubes in $\mathbb{R}^d$ using integer points in $\mathbb{Z}^d$. 
On the one hand, we obtain almost tight bounds on the competitive ratio in the lower dimension. On the other hand,  there is a significant gap between the lower and upper bound of the competitive ratio in higher-dimensional cases. We propose the following open problems.
\begin{enumerate}
    \item Can the lower bound result of unit balls be extended to any $d\in \mathbb{N}$?  
    \item Is there a lower bound on the competitive ratio for hitting unit hypercubes that match the upper bound of the problem? Is there any algorithm for hitting unit hypercubes with a competitive ratio of at most $O(d)$?
    \item There are small gaps between the lower and the upper bounds for unit balls and unit hypercubes in $\IR^2$ and $\IR^3$. We propose bridging these gaps as a future direction of research.
\end{enumerate}
 \bibliographystyle{plain} 
\bibliography{references}

\end{document}